\documentclass[english, letterpaper, 11pt]{article}

\usepackage{setspace}
\usepackage[T1]{fontenc}
\usepackage[utf8]{inputenc}
\usepackage{babel}
\usepackage{blindtext}
\usepackage{cite}
\usepackage[margin=1in]{geometry}
\usepackage{caption}
\usepackage[pdftex]{graphicx}
\graphicspath{{./fig/}}
\usepackage[cmex10]{amsmath}
\usepackage{amssymb}
\usepackage{array}
\usepackage[tight,footnotesize]{subfigure}
\usepackage{stfloats}
\usepackage{algorithm,algorithmic}
\usepackage{enumerate}
\usepackage{amsthm}
\usepackage{color}
\usepackage{mathrsfs}
\usepackage{multicol}
\usepackage{mathtools}
\usepackage{umoline}
\usepackage{arydshln}
\usepackage{hyperref}
\usepackage{authblk}

\newcommand{\defeq}{\vcentcolon=}
\newcommand{\eqdef}{=\vcentcolon}
\newcommand{\norm}[1]{\left\|{#1}\right\|}
\newcommand{\diag}{\operatorname{diag}}
\newcommand{\rank}{\operatorname{rank}}

\newcommand{\vect}{\operatorname{vec}}
\newcommand{\sgn}{\operatorname{sgn}}
\newcommand{\polylog}{\operatorname{polylog}}
\newcommand{\ind}[1]{\mathbf{1}\left(#1\right)}
\newcommand{\rmd}{\mathrm{d}}
\newcommand{\rmF}{\mathrm{F}}
\newcommand{\rmU}{\mathcal{U}}
\newcommand{\rmG}{\mathcal{G}}
\newcommand{\rmD}{\mathcal{D}}
\newcommand{\bfM}{\mathbf{M}}
\newcommand{\bbR}{\mathbb{R}}

\newcommand{\bbN}{\mathbb{N}}

\newcommand{\bbP}{\mathbb{P}}

\newcommand{\calX}{\mathcal{X}}

\newcommand{\calM}{\mathcal{M}}

\newcommand{\calU}{\mathcal{U}}
\newcommand{\calV}{\mathcal{V}}

\newcommand{\calA}{\mathcal{A}}
\newcommand{\calB}{\mathcal{B}}
\newcommand{\dB}{\Delta\calB}
\newcommand{\dX}{\Delta\calX}

\newtheorem{theorem}{Theorem}[section]
\newtheorem{lemma}[theorem]{Lemma}
\newtheorem{corollary}[theorem]{Corollary}
\newtheorem{proposition}[theorem]{Proposition}
\newtheorem{definition}[theorem]{Definition}

\begin{document}
\title{Optimal Sample Complexity for Stable Matrix Recovery\thanks{This work was supported in part by the National Science Foundation (NSF) under Grant IIS 14-47879. This paper was presented in part at ISIT 2016 \cite{Li2016}.}}
\author[1]{Yanjun Li}
\author[2]{Kiryung Lee}
\author[1]{Yoram Bresler}
\affil[1]{CSL and ECE, University of Illinois, Urbana-Champaign}
\affil[2]{ECE, Georgia Institute of Technology}
\date{}
\maketitle

\doublespacing

\abstract
Tremendous efforts have been made to study the theoretical and algorithmic aspects of sparse recovery and low-rank matrix recovery. This paper fills a theoretical gap in matrix recovery: the optimal sample complexity for stable recovery without constants or log factors. We treat sparsity, low-rankness, and potentially other parsimonious structures within the same framework: constraint sets that have small covering numbers or Minkowski dimensions. 
We consider three types of random measurement matrices (unstructured, rank-1, and symmetric rank-1 matrices), following probability distributions that satisfy some mild conditions. In all these cases, we prove a fundamental result -- the recovery of matrices with parsimonious structures, using an optimal (or near optimal) number of measurements, is stable with high probability.


\section{Introduction}

Matrix recovery plays a central role in many applications of signal processing and machine learning. It is widely known that an unknown matrix can be recovered from an underdetermined system of linear measurements, by exploiting parsimonious structures of the matrix, such as sparsity or low-rankness \cite{Eldar2012a,Davenport2016}. A special case where the unknown matrix is a sparse vector has been of particular interest in the context of compressed sensing and variable selection in linear regression.

Linear measurements of an unknown matrix are obtained through linear functionals, i.e., inner products with measurement matrices, which take different forms in different applications. In matrix completion \cite{Candes2009}, blind deconvolution via lifting \cite{Li2015d}, and bilinear regression \cite{Nzabanita2015}, the measure matrices have rank-1. In phase retrieval via lifting \cite{Candes2013a}, and in covariance matrix estimation via sketching \cite{Bahmani2015}, the measurement matrices are symmetric (or Hermitian) rank-1 matrices.

Given noise-free measurements, it is of interest to determine when the unknown matrix can be identified as the unique solution to an underdetermined system with a parsimonious prior. Sufficient conditions for the unique identification have been studied for recovery of low-rank and/or sparse matrices \cite{Donoho2003,Eldar2012}. As special cases with structured measurements, the uniqueness in bilinear inverse problems, especially blind deconvolution and blind calibration, is studied separately \cite{Li2015,Li2015b,Li2015d,Li2015e}. These results provided tight sample complexities for the exact recovery of the unknown matrix. 

However, in practice, measurements are corrupted with additive noise. It is therefore of interest to answer the question: under what conditions can the unknown matrix be estimated stably from noisy measurements. 
Many stability results have been shown by demonstrating the effectiveness of convex relaxation.
As for the recovery of sparse vectors, early results using the restricted isometry property (RIP) \cite{Candes2005,Candes2006} showed that stable recovery of $s$-sparse vectors of length $n$ is guaranteed with $m = O(s \log(n/s))$ i.i.d. Gaussian random measurements. 
Later RIPless analysis showed that, for a larger class of measurement functionals, $m = O(s \log(n/s))$ measurements are sufficient for stable recovery.
The results on stable recovery of sparse vectors were extended to the case of low-rank matrices \cite{Recht2010}, guaranteeing the recovery of $n\times n$ matrices of rank-$r$ from $m = O(r n \log n)$ linear measurements. 
Cand\`{e}s and Plan \cite{Candes2011} sharpened sample complexity to $m = O(r n)$. 
Chandrasekaran et al. unified the parsimonious models including low-rank matrices and sparse vectors as atomic sparsity models \cite{Chandrasekaran2012}. Using the Gaussian width of a tangent cone, they computed sample complexities for stable recovery that coincide with the empirical phase transition using convex relaxation. Recently, recovery of matrices that are sparse and low-rank has been studied (e.g., \cite{Lee2013}).
As for rank-1 measurement matrices, Cai and Zhang \cite{Cai2015} showed that stable recovery of $n_1\times n_2$ matrices of rank $r$ is achieved by $m = O(r (n_1 + n_2))$ measurements. Recently, stable recovery in blind deconvolution and phase retrieval \cite{Ahmed2014,Candes2013a,Lee2015a} has been studied by lifting to matrix recovery.

Another line of work studies the information-theoretic fundamental limit of sparse or low-rank matrix recovery, establishing the sample complexities achieved by an optimal decoder (practical or not). Wu and Verd\`{u} studied the performance of the optimal stable decoder for compressed sensing in a Bayesian framework \cite{Wu2012}. Also for compressed sensing, Reeves showed, without a prior distribution on the unknown sparse vector, the optimal sample complexity for stable recovery from i.i.d. Gaussian random measurements \cite{Reeves2014}. 

Riegler et al. studied the information-theoretic limit for the unique recovery of matrices in a set of small Minkowski dimension, using unstructured or rank-1 measurement matrices \cite{Riegler2015}. However, a relevant result on stable matrix recovery has been missing. Many key results in this paper build on the brilliant work by Riegler et al. \cite{Stotz2013,Riegler2015}. Our contributions include the following: (i) we refine the covering number argument used in \cite{Riegler2015} to achieve stability under the same sample complexity; (ii) we provide a simpler proof that gets rid of some unnecessary technicalities; (iii) we derive a concentration of measure bound with better constants for the case of uniformly distributed measurements treated by Stotz et al. \cite{Stotz2013}, and provide additional results for Gaussian random measurements. We provide more detailed comparisons later in the paper.

In this paper, we address the fundamental question of stable matrix recovery: how many measurements are sufficient to guarantee the existence of stable decoder? Similar to the paper by Riegler et al. \cite{Riegler2015}, our analysis covers a large category of problems, including compressed sensing, low-rank matrix recovery, phase retrieval, etc.

\section{Problem Statement}\label{sec:probstat}
\subsection{Notations}
The transpose of a matrix $A$ is denoted by $A^T$. The inner product of two matrices $A$ and $X$ are denoted by $\left<A,X\right>=\operatorname{trace}(A^TM)$. We use $\norm{\cdot}_0$ and $\norm{\cdot}_{\mathrm{r},0}$ to denote the numbers of nonzero entries and nonzero rows in a matrix, respectively. We use $\norm{\cdot}_2$ to denote the $\ell_2$ norm of a vector or the spectral norm of a matrix, and $\norm{\cdot}_\rmF$ to denote the Frobenious norm of a matrix. We use $[n]$ to denote the set of integers $\{1,2,\cdots,n\}$. If $J\subset [n]$, then the complement of $J$ is denoted by $J^c=[n]\backslash J$. We use $a^{(j)}$ to denote the $j$th entry of $a$, and $a^{(j_1:j_2)}$ to denote the subvector of $a$ consisting of the entries indexed by $j_1,j_1+1,\cdots,j_2$. Borrowing the colon notation from MATLAB, we use $X^{(J,:)}$ to denote the submatrix of $X$ consisting of the rows indexed by $J$. 

We use $\ind{\cdot}$ to denote the indicator function. Suppose $\Omega$ is the state space of a random variable $A$, and $E(A)$ is a statement about $A$ (also known as an event). Then $p_\rmD(\cdot)$ and $\bbP_\rmD[E(A)]$ denote the probability density function (PDF) of a distribution $\rmD$, and the probability of $E(A)$ when $A$ follows distribution $\rmD$. We have $\bbP_\rmD[E(A)] = \int_{\Omega} \ind{E(A)}\cdot p_\rmD(A) ~\rmd A$, which involves a minor abuse of notation -- the random variable and its value are both denoted by $A$.

We say a set $\Omega_\calX\in\bbR^{n_1\times n_2}$ is a cone, if for every $X\in\Omega_\calX$ and every $\sigma>0$, the scaled matrix $\sigma X\in \Omega_\calM$. The unit ball (with respect to the $\ell_2$ norm) in $\bbR^{n}$ centered at the origin is denoted by $\calB_{n}$. Then $x+R\calB_{n}$ denotes the ball in $\bbR^{n}$ of radius $R$ centered at $x$. Similarly, the unit ball (with respect to the Frobenius norm) in $\bbR^{n_1\times n_2}$ centered at the origin is denoted by $\calB_{n_1\times n_2}$. Then $X+R\calB_{n_1\times n_2}$ denotes the ball in $\bbR^{n_1\times n_2}$ of radius $R$ centered at $X$. We use $V_{n}=\int_{\calB_{n}} ~\rmd x$ to denote the volume of a unit ball in $\bbR^{n}$. Then the volume of a ball in $\bbR^n$ of radius $R$ is $R^nV_n$.

\subsection{Matrix Recovery}\label{sec:bdstat}
In this paper, we study the constrained matrix recovery (MR) problem. Suppose $X_0$ is an unknown $n_1\times n_2$ matrix. We have $m$ linear measurements of $X_0$, $y=\calA(X_0)+e\in\bbR^m$, where $\calA(X_0)$ is in the form of $\calA(X_0)=[\left<A_1,X_0\right>,\left<A_2,X_0\right>,\cdots,\left<A_m,X_0\right>]^T$, $A_1,A_2,\cdots,A_m\in\bbR^{n_1\times n_2}$ denote the measurement matrices, and $e= [e^{(1)},e^{(2)},\cdots,e^{(m)}]^T\in\bbR^m$ denotes the noise or other distortions in the measurement. The matrix recovery problem  refers to estimating the unknown matrix $X_0$ from $y$. We consider three models for the measurement matrices in this paper:
\begin{enumerate}
	\item Unstructured measurement matrices $\{A_j\}_{j=1}^m$.
	\item Rank-1 measurement matrices $\{A_j=a_j b_j^T\}_{j=1}^m$.
	\item Symmetric rank-1 measurement matrices $\{A_j=a_j a_j^T\}_{j=1}^m$, for which $n_1=n_2$.
\end{enumerate}
In this paper, we assume that the matrices $\{A_j\}_{j=1}^m$ (resp. vectors $\{a_j\}_{j=1}^m$, $\{b_j\}_{j=1}^m$) are i.i.d. random matrices (resp. vectors), following a probability distribution on $\bbR^{n_1\times n_2}$ (resp. $\bbR^{n_1}$, $\bbR^{n_2}$) that satisfies a mild concentration of measure inequality, which can be proved for a large category of probability distributions (e.g., uniform distribution on a ball, i.i.d. Gaussian distribution). More discussion is provided in Section \ref{sec:main}.

In matrix recovery, the number of measurements $m$ is often smaller than $n_1n_2$ -- the number of entries in $X_0$. For matrix recovery to be well-posed, the unknown matrix $X_0$ is assumed to belong to a known constraint set $\Omega_\calX\subset\bbR^{n_1\times n_2}$, which encodes our prior knowledge of $X_0$. As examples, we consider the following constraint sets:
\begin{enumerate}
	\item \emph{Matrices in a subspace.} The constraint set is a subspace of $\bbR^{n_1\times n_2},$ of dimension $t<n_1n_2$, which has an orthonormal basis $M_1, M_2,\cdots,M_t$. Then
\begin{equation}
\Omega_\calX = \{X \in\bbR^{n_1\times n_2}:\exists \beta\in\bbR^t,~\text{s.t.}~X = \sum_{i=1}^{t}\beta^{(i)} M_i\}. \label{eq:cssp}
\end{equation}
Examples of such subspaces include the sets of Hankel matrices, Toeplitz matrices, and symmetric matrices. Hankel (resp. Toeplitz) matrices, in which each skew-diagonal (resp. diagonal) is constant, i.e., $X^{(j,k)}=X^{(j+1,k-1)}$ (resp. $X^{(j,k)}=X^{(j+1,k+1)}$), reside in a subspace of dimension $t = n_1+n_2-1$. Symmetric matrices, which are square matrices equal to their transposes, i.e., $n_1=n_2=n$ and $X^{(j,k)}=X^{(k,j)}$, reside in a subspace of dimension $t = n(n+1)/2$. Symmetric Toeplitz matrices reside in a subspace of dimension $t = n$.

	\item \emph{Sparse matrices.} The constraint set is the set of $s$-sparse matrices over a dictionary, whose atoms are $M_1, M_2,\cdots,M_t$. Then
\begin{equation}
\Omega_\calX = \{X\in\bbR^{n_1\times n_2}: \exists \beta\in\bbR^t,~\text{s.t.}~\norm{\beta}_0\leq s,~X = \sum_{i=1}^{t}\beta^{(i)} M_i\}. \label{eq:css}
\end{equation}
When $\{M_i\}_{i=1}^{t}$ are symmetric matrices, the constraint set is the set of sparse symmetric matrices. When $n_2 = 1$, the sparse matrix recovery problem reduces to sparse vector recovery.

The following notations will be used in Section \ref{sec:cn}. Let $\bfM=[\vect(M_1),\vect(M_2),\cdots,\vect(M_t)]$, then we have $\vect(X)=\bfM\beta$. Define
\begin{align*}
\sigma_{s,\min} = \underset{\norm{\beta}_2 = 1,\norm{\beta}_0\leq s}{\min} \norm{\bfM\beta}_2,\qquad \sigma_{s,\max} = \underset{\norm{\beta}_2 = 1,\norm{\beta}_0\leq s}{\max} \norm{\bfM\beta}_2,\qquad \kappa_s = \frac{\sigma_{s,\max}}{\sigma_{s,\min}}. 
\end{align*}
For example, if $\bfM$ is an orthonormal basis (e.g., the standard basis),  then $\kappa_s = \sigma_{s,\min}=\sigma_{s,\max}=1$. If $\bfM$ has a restricted isometry constant $\delta_s$ \cite{Candes2005}, then $\sigma_{s,\min} \geq \sqrt{1-\delta_s}$, $\sigma_{s,\max} \leq \sqrt{1+\delta_s}$, and $\kappa_s \leq \sqrt{\frac{1+\delta_s}{1-\delta_s}}$. In this paper, we assume that $\sigma_{4s,\min}>0$ and hence $\kappa_{4s}<\infty$. 
	\item \emph{Low-rank matrices.} The constraint set is the set of matrices of rank at most $r$, i.e., 
\begin{equation}
\Omega_\calX = \{X\in\bbR^{n_1\times n_2}: \rank(X)\leq r\}. \label{eq:csl}
\end{equation}
	\item \emph{Sparse low-rank matrices.} We consider the special set of matrices that have at most rank $r$, have at most $s_1$ nonzero rows, and have at most $s_2$ nonzero columns ($r<\min\{s_1,s_2\}$). The constraint set is
\begin{equation}
\Omega_\calX = \{X \in\bbR^{n_1\times n_2}: \rank(X)\leq r, \norm{X}_{\mathrm{r},0}\leq s_1, \norm{X^T}_{\mathrm{r},0}\leq s_2\}. \label{eq:cssl}
\end{equation}
	\item \emph{Symmetric low-rank matrices.} Symmetry can be combined with low-rank structures in \eqref{eq:csl} and \eqref{eq:cssl}, the results of which are:
\begin{align}
\Omega_\calX =& \{X\in\bbR^{n\times n}: X=X^T,~\rank(X)\leq r\}. \label{eq:csl2}\\
\Omega_\calX =& \{X \in\bbR^{n\times n}: X=X^T,~\rank(X)\leq r,~\norm{X}_{\mathrm{r},0}\leq s\}. \label{eq:cssl2}
\end{align}
\end{enumerate}  

Note that all the above constraint sets are cones. For all practical purposes, the matrix $X_0$ has finite energy. Hence it suffices recover $X_0$ subject to the constraint set restricted to a ball, whose radius is sufficiently large. 
We define the following shorthand notations for the rest of this paper:
\begin{align}
\Omega_\calB \defeq & \Omega_\calX \bigcap \calB_{n_1\times n_2}, \label{eq:setB}\\
\Omega_{\dB} \defeq & \Omega_\calB-\Omega_\calB = (\Omega_\calX \bigcap \calB_{n_1\times n_2})-(\Omega_\calX \bigcap \calB_{n_1\times n_2}),  \label{eq:setdB} \\
\Omega_{\dX} \defeq & (\Omega_\calX-\Omega_\calX)\bigcap \calB_{n_1\times n_2}.  \label{eq:setdX}
\end{align}
Then we can estimate $X_0$, for example, by solving the following constrained least squares problem:
\begin{align*}
\text{(MR)}\qquad\underset{X}{\min.}~~& \norm{\calA(X)-y}_2,\\
\text{s.t.}~~& X\in \sigma\Omega_\calB.
\end{align*}
If the radius $\sigma<\infty$, (MR) has a bounded constraint set
\[
\sigma\Omega_\calB = \Omega_\calX \bigcap \sigma\calB_{n_1\times n_2} = \{X\in\Omega_\calX: \norm{X}_\rmF\leq \sigma\}.
\]
If $\sigma = \infty$, the constraint set becomes $\Omega_\calX$, which is unbounded.

\subsection{Stability}

We introduce the following notions of stability:
\begin{definition}~
\begin{enumerate}
	\item \textbf{Single point stability:} We say that the recovery of $X_0\in\sigma\Omega_\calB$ using measurement operator $\calA$ is stable at level $(\delta,\varepsilon)$, if for all $X\in\sigma\Omega_\calB$ such that $\norm{\calA(X)-\calA(X_0)}_2\leq \delta$, we have $\norm{X-X_0}_\calX\leq \varepsilon$.
	\item \textbf{Uniform stability:} We say that the recovery on $\sigma\Omega_\calB$ using measurement operator $\calA$ is uniformly stable at level $(\delta,\varepsilon)$, if for all $X_1,X_2\in \sigma\Omega_\calB$ such that $\norm{\calA(X_1)-\calA(X_2)}_2\leq \delta$, we have $\norm{X_1-X_2}_\calX\leq \varepsilon$.
\end{enumerate}
In both definitions, $\norm{\cdot}_\calX$ can either be the Frobenius norm $\norm{\cdot}_\rmF$ or the spectral norm $\norm{\cdot}_2$, and $\varepsilon = \varepsilon(\delta)$ is a function of $\delta$ that vanishes as $\delta$ approaches $0$.
\end{definition}

The stability, as defined above, would guarantee the accuracy of the constrained least squares estimation. Let $X_1$ denote the solution to (MR) with noisy measurement. Suppose the perturbation in the measurement is small, $\norm{e}_2\leq \frac{\delta}{2}$ for some small $\delta >0$. Then the deviation of $\calA(X_1)$ from $\calA(X_0)$ is small, i.e.,
\begin{align*}
\norm{\calA(X_1)-\calA(X_0)}_2 \leq \norm{\calA(X_1)-y}_2+\norm{\calA(X_0)-y}_2 \leq 2\norm{\calA(X_0)-y}_2 = 2\norm{e}_2 \leq \delta.
\end{align*}
By the definition of single point stability or uniform stability, we have $\norm{X_1-X_0}_\calX\leq \varepsilon$, which is also a small quantity.

If the recovery of $X_0$ is stable, then for every $\varepsilon>0$, there exists $\delta>0$ such that for all $X\in\sigma\Omega_\calB$ that satisfies $\norm{\calA(X)-\calA(X_0)}_2\leq \delta$, we have $\norm{X-X_0}_\calX\leq \varepsilon$. If the recovery of all matrices in $\sigma\Omega_\calB$ is stable, then for every $\varepsilon>0$, there exists $\delta>0$ such that for all $X_1,X_2\in\sigma\Omega_\calB$ that satisfies $\norm{\calA(X_1)-\calA(X_2)}_2\leq \delta$, we have $\norm{X_1-X_2}_\calX\leq \varepsilon$. If $\calA$ (restricted to the domain $\sigma\Omega_\calB$) is invertible, i.e., there exists $\calA^{-1}: \calA(\sigma\Omega_\calB)\rightarrow \sigma\Omega_\calB$, then \emph{single point stability} at $X_0$ implies that $\calA^{-1}$ is continuous at $\calA(X_0)$; \emph{uniform stability} on $\sigma\Omega_\calB$ implies that $\calA^{-1}$ is uniformly continuous on $\calA(\sigma\Omega_\calB)$.

Suppose $\Omega_\calX$ is a cone, and we need to evaluate the stability on a bounded constraint set $\sigma\Omega_\calB$ ($\sigma<\infty$). We can scale $X_0$ and the radius of the ball by $\frac{1}{\sigma}$ simultaneously. If for all $X\in\Omega_\calB$ such that $\norm{\calA(X)-\calA(\frac{X_0}{\sigma})}_2\leq \delta$, we have $\norm{X-\frac{X_0}{\sigma}}_\calX\leq \varepsilon(\delta)$, then for all $X\in\sigma\Omega_\calB$ such that $\norm{\calA(X)-\calA(X_0)}_2\leq \delta$, we have $\norm{X-X_0}_\calX\leq \sigma\varepsilon(\frac{\delta}{\sigma})$. In other words, stability on $\Omega_\calB$ implies stability on any bounded subset of $\Omega_\calX$. Therefore, in this paper, we consider $\Omega_\calB$ and $\Omega_\calX$ as representatives for bounded and unbounded constraint sets. The main results bound the probability of three events:
\begin{enumerate}
	\item Single point stability on bounded constraint set $\Omega_\calB$.
	\item Uniform stability on bounded constraint set $\Omega_\calB$.
	\item Uniform stability on unbounded constraint set $\Omega_\calX$.
\end{enumerate}

\subsection{Modeling Error}
In practice, the true matrix $X_0$ may not belong to the constraint set $\Omega_\calX$, but may be close to it. Let $\widehat{X}_0 = \arg\min_{X\in\Omega_\calX} \norm{X-X_0}_\calX$ denote the projection of $X_0$ onto $\Omega_\calX$, and suppose we have the following bounds on the modeling error $\norm{X_0-\widehat{X}_0}_\calX$ and the operator norm of $\calA$:
\[
\norm{X_0-\widehat{X}_0}_\calX = \underset{X\in\Omega_\calX}{\min} \norm{X-X_0}_\calX\leq \varepsilon_M,
\]
\[
\norm{\calA}_{\calX\rightarrow 2} = \max\limits_{X\in\bbR^{n_1\times n_2}, \norm{X}_\calX = 1} \norm{\calA(X)}_2 \leq L.
\]
Then the error in the estimator $X_1$ is bounded by
\begin{equation}
\norm{X_1-X_0}_\calX \leq \norm{X_1-\widehat{X}_0}_\calX+\norm{X_0-\widehat{X}_0}_\calX \leq \varepsilon(2L\varepsilon_M + \delta) + \varepsilon_M, \label{eq:model_error}
\end{equation}
where the bound on the first term follows from the stability of matrix recovery at $\widehat{X}_0\in\Omega_\calX$ (the recovery error $\varepsilon(\cdot)$ is a function of the measure error), and from the following bound on the measurement error:
\begin{align*}
\norm{\calA(X_1)-\calA(\widehat{X}_0)}_2 \leq & \norm{\calA(X_1)-y}_2 + \norm{\calA(\widehat{X}_0)-y}_2 \\
\leq & 2\norm{\calA(\widehat{X}_0)-y}_2 \\
\leq & 2\norm{\calA(\widehat{X}_0)-\calA(X_0)}_2 + 2\norm{\calA(X_0)-y}_2 \\
\leq & 2\norm{\calA}_{\calX\rightarrow 2}\norm{\widehat{X}_0-X_0}_\calX + 2\norm{e}_2 \\
\leq & 2L\varepsilon_M + 2\times\frac{\delta}{2} = 2L\varepsilon_M + \delta.
\end{align*}
The first and third lines follow from triangle inequality, and the second line follows from the optimality of $X_1$ in (MR).

By \eqref{eq:model_error}, even in the presence of modeling error, stability of recovery can guarantee that the recovery error is bounded by a small quantity that is a function of the modeling error ($\varepsilon_M$) and the measurement error ($\delta$).


\section{Covering Number and Minkowski Dimension} \label{sec:cnmd}
The conditions for stability of the matrix recovery problem (MR) are expressed in terms of the covering number or the Minkowski dimension of the constraint set $\Omega_\calX$, which are defined as follows.
\begin{definition}\label{def:minkowski}
The lower and upper Minkowski dimensions of a nonempty bounded set $\Omega\subset\bbR^{n_1\times n_2}$ are
\[
\underline{\dim}_\mathrm{B}(\Omega)\eqdef \underset{\rho\rightarrow 0}{\lim\inf}\frac{\log N_{\Omega}(\rho)}{\log\frac{1}{\rho}},\qquad
\overline{\dim}_\mathrm{B}(\Omega)\eqdef \underset{\rho\rightarrow 0}{\lim\sup}\frac{\log N_{\Omega}(\rho)}{\log\frac{1}{\rho}},
\]
where $N_{\Omega}(\rho)$ denotes the covering number of set $\Omega$ given by
\[
N_{\Omega}(\rho) = \min\biggl\{k\in\bbN: \Omega\subset \underset{i\in\{1,2,\cdots,k\}}{\bigcup}(X_i+\rho\calB_{n_1\times n_2}),~X_i\in\bbR^{n_1\times n_2}\biggr\}.
\]
If $\underline{\dim}_\mathrm{B}(\Omega) = \overline{\dim}_\mathrm{B}(\Omega)$, then it is simply called the Minkowski dimension, denoted by $\dim_\mathrm{B}(\Omega)$.
\end{definition}

The covering number of a set characterizes its description complexity. As will be shown in Section \ref{sec:main}, bounds on the covering numbers of $\Omega_\calB$ play an important role in the sample complexities of matrix recovery problems that guarantee single point stability. Similarly, sample complexities that guarantee uniform stability are expressed in terms of bounds on the covering numbers of $\Omega_{\dB}$ and $\Omega_{\dX}$, defined by \eqref{eq:setdB} and \eqref{eq:setdX}, respectively. 

\subsection{Bounds on Covering Numbers}\label{sec:cn}
We prove in Appendix \ref{app:cn} the following bounds on the covering numbers of the constraint sets defined in Section \ref{sec:bdstat}:
\begin{proposition}\label{pro:cnB}
If $\Omega_\calX$ is \eqref{eq:cssp} -- \eqref{eq:cssl2}, then the covering number of $\Omega_\calB$, defined by \eqref{eq:setB}, satisfies $N_{\Omega_\calB}(\rho) \leq C_1\left(\frac{1}{\rho}\right)^{d_1}$ for all $0<\rho<1$, where $d_1$ and $C_1$ are global constants that only depend on $n_1$, $n_2$, $s$, $r$, $s_1$, and $s_2$. The expressions for $d_1$ and $C_1$ are summarized in Table \ref{tab:cnB}.
\end{proposition}

\begin{table}[htbp]%
\renewcommand{\arraystretch}{1.3}
\begin{center}
\begin{tabular}{ >{\raggedright\arraybackslash}m{2.5in} >{\centering\arraybackslash}m{0.8in} >{\centering\arraybackslash}m{1.8in} >{\centering\arraybackslash}m{0in} }
\hline
\multicolumn{1}{c}{$\Omega_\calX$} & $d_1$ & $C_1$ &\\
\hline
\eqref{eq:cssp}: $t$-dimensional subspace  & $t$ & $3^t$ &\\
\eqref{eq:css}: $s$-sparse matrices  & $s$ & $(3\kappa_{2s})^s\cdot {t\choose s}$ & \\
\eqref{eq:csl}: rank-$r$ matrices  & $(n_1+n_2)r$ & $\left(6\sqrt{r}\right)^{(n_1+n_2)r}$ & \\
\eqref{eq:cssl}: sparse rank-$r$ matrices  & $(s_1+s_2)r$ & $\left(6\sqrt{r}\right)^{(s_1+s_2)r}\cdot{n_1 \choose s_1}{n_2 \choose s_2}$ & \\
\eqref{eq:csl2}: symmetric rank-$r$ matrices  & $nr$ & $(r+1)\left(6\sqrt{r}\right)^{nr}$ & \\
\eqref{eq:cssl2}: symmetric sparse rank-$r$ matrices  & $sr$ & $(r+1)\left(6\sqrt{r}\right)^{sr}\cdot {n\choose s}$ & \\
\hline
\end{tabular}
\end{center}
\caption{A summary of the constants in Proposition \ref{pro:cnB} .}
\label{tab:cnB}
\end{table}

\begin{proposition}\label{pro:cndB}
If $\Omega_\calX$ is \eqref{eq:cssp} -- \eqref{eq:cssl2}, then the covering number of the difference set $\Omega_{\dB}=\Omega_\calB-\Omega_\calB$, defined by \eqref{eq:setdB}, satisfies $N_{\Omega_{\dB}}(\rho) \leq C_2\left(\frac{1}{\rho}\right)^{d_2}$ for all $0<\rho<1$, where $d_2$ and $C_2$ are global constants that only depend on $n_1$, $n_2$, $s$, $r$, $s_1$, and $s_2$. The expressions for $d_2$ and $C_2$ are summarized in Table \ref{tab:cndB}.
\end{proposition}
\begin{table}[htbp]%
\renewcommand{\arraystretch}{1.3}
\begin{center}
\begin{tabular}{ >{\raggedright\arraybackslash}m{2.5in} >{\centering\arraybackslash}m{0.8in} >{\centering\arraybackslash}m{1.8in} >{\centering\arraybackslash}m{0in} }
\hline
\multicolumn{1}{c}{$\Omega_\calX$} & $d_2$ & $C_2$ &\\
\hline
\eqref{eq:cssp}: $t$-dimensional subspace & $t$ & $6^t$ &\\
\eqref{eq:css}: $s$-sparse matrices & $2s$ & $(6\kappa_{2s})^{2s}\cdot {t\choose s}^2$ &\\
\eqref{eq:csl}: rank-$r$ matrices & $2(n_1+n_2)r$ & $\left(12\sqrt{r}\right)^{2(n_1+n_2)r}$ &\\
\eqref{eq:cssl}: sparse rank-$r$ matrices & $2(s_1+s_2)r$ & $\left(12\sqrt{r}\right)^{2(s_1+s_2)r}\cdot{n_1 \choose s_1}^2{n_2 \choose s_2}^2$ &\\
\eqref{eq:csl2}: symmetric rank-$r$ matrices & $2nr$ & $(r+1)^2\left(12\sqrt{r}\right)^{2nr}$ & \\
\eqref{eq:cssl2}: symmetric sparse rank-$r$ matrices & $2sr$ & $(r+1)^2\left(12\sqrt{r}\right)^{2sr}\cdot {n\choose s}^2$ & \\
\hline
\end{tabular}
\end{center}
\caption{A summary of the constants in Proposition \ref{pro:cndB} .}
\label{tab:cndB}
\end{table}

\begin{proposition}\label{pro:cndX}
If $\Omega_\calX$ is \eqref{eq:cssp} -- \eqref{eq:cssl2}, then the covering number of the difference set $\Omega_{\dX}=\Omega_\calX-\Omega_\calX$, defined by \eqref{eq:setdX}, satisfies $N_{\Omega_{\dX}}(\rho) \leq C_3\left(\frac{1}{\rho}\right)^{d_3}$ for all $0<\rho<1$, where $d_3$ and $C_3$ are global constants that only depend on $n_1$, $n_2$, $s$, $r$, $s_1$, and $s_2$. The expressions for $d_3$ and $C_3$ are summarized in Table \ref{tab:cndX}.
\end{proposition}
\begin{table}[htbp]%
\renewcommand{\arraystretch}{1.3}
\begin{center}
\begin{tabular}{ >{\raggedright\arraybackslash}m{2.5in} >{\centering\arraybackslash}m{0.8in} >{\centering\arraybackslash}m{1.8in} >{\centering\arraybackslash}m{0in} }
\hline
\multicolumn{1}{c}{$\Omega_\calX$} & $d_3$ & $C_3$ &\\
\hline
\eqref{eq:cssp}: $t$-dimensional subspace & $t$ & $3^t$ &\\
\eqref{eq:css}: $s$-sparse matrices & $2s$ & $(3\kappa_{4s})^{2s}\cdot {t\choose 2s}$ &\\
\eqref{eq:csl}: rank-$r$ matrices & $2(n_1+n_2)r$ & $\left(6\sqrt{2r}\right)^{2(n_1+n_2)r}$ &\\
\eqref{eq:cssl}: sparse rank-$r$ matrices & $4(s_1+s_2)r$ & $\left(6\sqrt{2r}\right)^{4(s_1+s_2)r}\cdot{n_1 \choose 2s_1}{n_2 \choose 2s_2}$ &\\
\eqref{eq:csl2}: symmetric rank-$r$ matrices & $2nr$ & $(2r+1)\left(6\sqrt{2r}\right)^{2nr}$ & \\
\eqref{eq:cssl2}: symmetric sparse rank-$r$ matrices & $4sr$ & $(2r+1)\left(6\sqrt{2r}\right)^{4sr}\cdot {n\choose 2s}$ & \\
\hline
\end{tabular}
\end{center}
\caption{A summary of the constants in Proposition \ref{pro:cndX} .}
\label{tab:cndX}
\end{table}

\subsection{Alternative Bounds Using Minkowski Dimensions}\label{sec:cn_alt}
Given the bounds on the covering numbers in Proposition \ref{pro:cnB}, the upper Minkowski dimensions of the three constraint sets (sparse matrices, low-rank matrices, and sparse low-rank matrices), are bounded by $s$, $(n_1+n_2)r$, and $(s_1+s_2)r$, respectively. On the other hand, if we are given a bound on the upper Minkowski dimension of a set, we can bound its covering number.
\begin{proposition}\label{pro:cnmd}
If $\overline{\dim}_\mathrm{B}(\Omega)\leq d$, then there exists $\rho_0>0$, such that
\[
N_{\Omega}(\rho) \leq \left(\frac{1}{\rho}\right)^{d+1},\quad \forall~ 0<\rho < \rho_0.
\]
\end{proposition}

Combining Proposition \ref{pro:cnmd} with bounds on the Minkowski dimensions of the sets $\Omega_\calB$, $\Omega_{\dB}$, and $\Omega_{\dX}$ derived in Appendix \ref{app:cn}, we have the following alternative bounds for the covering numbers.
\begin{corollary}\label{cor:cn_alt}
If $\Omega_\calX$ is the set of low-rank matrices \eqref{eq:csl} or sparse low-rank matrices \eqref{eq:cssl}, then the covering numbers of $\Omega_\calB$, $\Omega_{\dB}$, and $\Omega_{\dX}$ satisfy:
\begin{enumerate}
	\item There exists $\rho_1>0$, such that $N_{\Omega_\calB}(\rho) \leq \left(\frac{1}{\rho}\right)^{d_1}$ for all $0<\rho<\rho_1$.
	\item There exists $\rho_2>0$, such that $N_{\Omega_{\dB}}(\rho) \leq \left(\frac{1}{\rho}\right)^{d_2}$ for all $0<\rho<\rho_2$.
	\item There exists $\rho_3>0$, such that $N_{\Omega_{\dX}}(\rho) \leq \left(\frac{1}{\rho}\right)^{d_3}$ for all $0<\rho<\rho_3$.
\end{enumerate}
The expressions for $d_1$, $d_2$, and $d_3$ are summarized in Table \ref{tab:cn_alt}.
\end{corollary}
\begin{table}[htbp]%
\renewcommand{\arraystretch}{1.3}
\begin{center}
\begin{tabular}{ >{\raggedright\arraybackslash}m{1.8in} >{\centering\arraybackslash}m{1.3in} >{\centering\arraybackslash}m{1.3in} >{\centering\arraybackslash}m{1.3in} >{\centering\arraybackslash}m{0in} }
\hline
\multicolumn{1}{c}{$\Omega_\calX$} & $d_1$ & $d_2$ & $d_3$ &\\
\hline
\eqref{eq:csl}: rank-$r$ matrices & $(n_1+n_2-r)r+1$ & $2(n_1+n_2-r)r+1$ & $2(n_1+n_2-2r)r+1$ &\\
\eqref{eq:cssl}: sparse rank-$r$ matrices & $(s_1+s_2-r)r+1$ & $2(s_1+s_2-r)r+1$ & $4(s_1+s_2-r)r+1$ &\\
\hline
\end{tabular}
\end{center}
\caption{A summary of the constants in Corollary \ref{cor:cn_alt} .}
\label{tab:cn_alt}
\end{table}

The bounds on the covering numbers of the sets of low-rank (or sparse low-rank) matrices in Corollary \ref{cor:cn_alt} are sharper than those in Propositions \ref{pro:cnB} -- \ref{pro:cndX}, the proofs of which are simplified by relaxation. However, the bounds in Corollary \ref{cor:cn_alt} hold only for sufficiently small $\rho$, whereas the bounds in Propositions \ref{pro:cnB} -- \ref{pro:cndX} hold for any $\rho>0$. In general, bounding the covering number via Minkowski dimension is unnecessary, if one can directly obtain a covering number bound whose exponent matches the number of degrees of freedom (e.g., the bounds for a subspace or a set of sparse matrices in Propositions \ref{pro:cnB} -- \ref{pro:cndX}).

\section{Main Results}\label{sec:main}


\subsection{Unstructured Measurement Matrices}\label{sec:unstructured}

Riegler et al. \cite{Riegler2015} showed that the matrix recovery problem has a unique solution if the number $m$ of linear measurements is greater than the lower Minkowski dimension of the constraint set, which implies, for example, that $m>(n_1+n_2)r$ is sufficient to guarantee the uniqueness of the solution when the constraint set is defined by \eqref{eq:csl}. In this section, we show that for random measurement matrices that follow certain distributions, the same sample complexity can also guarantee stability, with high probability. 

The stability results in Theorem \ref{thm:stability} hold under assumptions (A1) and (A2), on the constraint set and the measurement matrices, respectively:
\begin{enumerate}
	\item[(A1)] The constraint $\Omega_\calX$ satisfies that the covering numbers $N_1(\rho)$, $N_2(\rho)$, $N_3(\rho)$ of sets $\Omega_\calB$, $\Omega_{\dB}$ and $\Omega_{\dX}$ (defined by \eqref{eq:setB}, \eqref{eq:setdB}, and \eqref{eq:setdX}, respectively) are bounded by
\begin{align}
N_i(\rho) \leq C_i\left(\frac{1}{\rho}\right)^{d_i},\qquad \forall \rho < \rho_i,~ i = 1,2,3,  \label{eq:cn_bound}
\end{align}
where $\rho_i>0$, and $C_i$ is independent of $\rho$.
	\item[(A2)] The measurement matrices $\{A_j\}_{j=1}^m\subset \bbR^{n_1\times n_2}$ are i.i.d. random matrices following a distribution $\rmD$ that satisfies the following concentration of measure bounds ($\varepsilon,\delta>0$):
\begin{align}
&\bbP_\rmD\left[\norm{A}_\rmF \leq R,~\left|\left<A,X\right>\right|\leq \delta \right] \leq C_{\rmD,R}\cdot\frac{\delta}{\varepsilon},\qquad \forall X~~\text{s.t.}~\norm{X}_\rmF \geq \varepsilon, \label{eq:concentration1} \\
&\bbP_{\rmD} [\norm{A}_\rmF > R] = \theta_{\rmD,R}. \label{eq:def_theta1}
\end{align}
The constant $C_{\rmD,R}$ depends on distribution $\rmD$ and radius $R$, but not on $\varepsilon,\delta$.
\end{enumerate}

\begin{theorem}\label{thm:stability}
Suppose the constraint set and the measurement matrices satisfy assumptions (A1) and (A2), respectively. If $m>d_i$, $\varepsilon<1$, and $\delta < R\rho_i$, then with probability $1-P_i$, where
\begin{equation}
P_i \leq C_i\left(3C_{\rmD,R}\right)^m\cdot R^{d_i} \cdot\frac{\delta^{m-d_i}}{\varepsilon^m} + m\cdot \theta_{\rmD,R},  \label{eq:prob_bd}
\end{equation}
we have the following:
\begin{enumerate}
	\item for $i=1$, single point stability on bounded constraint set $\Omega_\calB$. 
	\item for $i=2$, uniform stability on bounded constraint set $\Omega_\calB$.
	\item for $i=3$, uniform stability on unbounded constraint set $\Omega_\calX$. 
\end{enumerate}
In all three cases, the norm $\norm{\cdot}_\calX$ in which the recovery error of the matrix is measured in the definition of stability, is the \emph{Frobenius norm} $\norm{\cdot}_\rmF$.
\end{theorem}

Since $\theta_{\rmD,R}$ is non-negative and non-increasing in $R$, it converges to its infimum $0$ as $R$ approaches infinity. Therefore, one can always choose a sufficiently large $R$ such that $\theta_{\rmD,R} < \frac{1}{m}$. To make sure that the probability in \eqref{eq:prob_bd} is non-trivial, let
\begin{align}
\varepsilon=\varepsilon(\delta) > 3C_{\rmD,R} \cdot\left(\frac{C_i}{1-m\cdot \theta_{\rmD,R}}\right)^\frac{1}{m}\cdot R^\frac{d_i}{m}\cdot \delta^{1-\frac{d_i}{m}}. \label{eq:vanish}
\end{align}
Fixing $R$, the right hand side of \eqref{eq:vanish} is a function of $\delta$ that vanishes as $\delta$ approaches $0$, which meets the definition of stability.

Next, we specialize Theorem \ref{thm:stability} for the cases of uniform distribution and of i.i.d. Gaussian distribution. The proof of Theorem \ref{thm:stability} follows the same steps as \cite[Proposition 1]{Stotz2013}. However, we refine the argument by covering the relevant set with balls of a different radius, so that we can show stability with the same number of measurements. We also generalize the results to cover measurement models following other distributions, e.g., Gaussian distribution in Corollary \ref{cor:gaussian}. For uniformly distributed measurements, our guarantee Corollary \ref{cor:uniform} and Lemma \ref{lem:concentration1} has improved constants compared with previous results \cite[Lemma 3]{Stotz2013}.
\begin{corollary}\label{cor:uniform}
Suppose the constraint set satisfies assumption (A1), and the measurement matrices $\{A_j\}_{j=1}^m$ are i.i.d. random matrices following distribution $\rmU$ -- uniform distribution on the ball $R\calB_{n_1n_2}$. Then the stability results in Theorem \ref{thm:stability} hold, except for a small probability:
\begin{align*}
P_i \leq C_i\left(\frac{6\cdot V_{n_1n_2-1}}{V_{n_1n_2}}\right)^m\left(\frac{\delta}{R}\right)^{m-d_i} \left(\frac{1}{\varepsilon}\right)^m. 
\end{align*}
\end{corollary}

\begin{corollary}\label{cor:gaussian}
Suppose the constraint set satisfies assumption (A1), and the measurement matrices $\{A_j\}_{j=1}^m$ are i.i.d. random matrices following distribution $\rmG$ -- the entries of the measurement matrices are i.i.d. Gaussian random variables $N(0,\sigma^2)$. Then the stability results in Theorem \ref{thm:stability} hold, except for a small probability:
\begin{align*}
P_i \leq C_i\left(\frac{3\sqrt{2}}{\sqrt{\pi}\sigma}\right)^m\cdot R^{d_i}\cdot \frac{\delta^{m-d_i}}{\varepsilon^m} + m\cdot e^{-\frac{n_1n_2}{2}\left(\frac{R^2}{n_1n_2\sigma^2}-1-\ln \frac{R^2}{n_1n_2\sigma^2} \right)}, \qquad \forall R>\sqrt{n_1n_2}\sigma. 
\end{align*}
\end{corollary}

Combining the above results with the bounds on covering numbers in Section \ref{sec:cnmd}, we have Corollary \ref{cor:mr}.
\begin{corollary}\label{cor:mr} 
The stability results in Theorem \ref{thm:stability}, Corollary \ref{cor:uniform}, and Corollary \ref{cor:gaussian} hold
\begin{enumerate}
	\item for $\Omega_\calX$ defined by \eqref{eq:cssp} -- \eqref{eq:cssl}, under the sample complexities in Table \ref{tab:sc}.
	\item for $\Omega_\calX$ defined by \eqref{eq:csl} or \eqref{eq:cssl} when perturbations are small ($\delta<R\rho_i$, $i=1,2,3$), under the less demanding sample complexities in Table \ref{tab:sc_alt}.
\end{enumerate}
\end{corollary}
The first result in Corollary \ref{cor:mr} follows from  the bounds on covering numbers in Section \ref{sec:cn}, and the second result follows from the alternative bounds in Section \ref{sec:cn_alt}. 

\begin{table}[htbp]%
\renewcommand{\arraystretch}{1.3}
\begin{center}
\begin{tabular}{ >{\centering\arraybackslash}m{0.5in} >{\centering\arraybackslash}m{1.5in} >{\centering\arraybackslash}m{1.5in} >{\centering\arraybackslash}m{1.5in} >{\centering\arraybackslash}m{0in} }
\hline
$\Omega_\calX$ & Single point stability on $\Omega_\calB$ & Uniform stability on $\Omega_\calB$ & Uniform stability on $\Omega_\calX$ &\\
\hline
\eqref{eq:cssp} & $m>t$ & $m>t$ & $m>t$ &\\
\eqref{eq:css} & $m>s$ & $m>2s$ & $m>2s$ &\\
\eqref{eq:csl} & $m>(n_1+n_2)r$ & $m>2(n_1+n_2)r$ & $m>2(n_1+n_2)r$ &\\
\eqref{eq:cssl} & $m>(s_1+s_2)r$ & $m>2(s_1+s_2)r$ & $m>4(s_1+s_2)r$ &\\
\hline
\end{tabular}
\end{center}
\caption{A summary of sample complexities for stable recovery.}
\label{tab:sc}
\end{table}

\begin{table}[htbp]%
\renewcommand{\arraystretch}{1.3}
\begin{center}
\begin{tabular}{ >{\centering\arraybackslash}m{0.5in} >{\centering\arraybackslash}m{1.5in} >{\centering\arraybackslash}m{1.5in} >{\centering\arraybackslash}m{1.5in} >{\centering\arraybackslash}m{0in} }
\hline
$\Omega_\calX$ & Single point stability on $\Omega_\calB$ ($\delta<R\rho_1$) & Uniform stability on $\Omega_\calB$ ($\delta<R\rho_2$) & Uniform stability on $\Omega_\calX$ ($\delta<R\rho_3$) &\\
\hline
\eqref{eq:csl} & $m>(n_1+n_2-r)r+1$ & $m>2(n_1+n_2-r)r+1$ & $m>2(n_1+n_2-2r)r+1$ &\\
\eqref{eq:cssl} & $m>(s_1+s_2-r)r+1$ & $m>2(s_1+s_2-r)r+1$ & $m>4(s_1+s_2-r)r+1$ &\\
\hline
\end{tabular}
\end{center}
\caption{A summary of sample complexities for stable recovery against small perturbations.}
\label{tab:sc_alt}
\end{table}



\subsection{Rank-1 Measurement Matrices}\label{sec:rank1}

Next, Theorem \ref{thm:rank1_stability} shows that the same sample complexities as in Theorem \ref{thm:stability} apply to matrix recovery with rank-1 measurement matrices. 

In this section, the measurement matrices have the form $A_j=a_j b_j^T$. The distribution of random matrix $A_j$ is described in terms of the distributions of random vectors $a_j\in\bbR^{n_1}$ and $b_j\in\bbR^{n_2}$. Theorem \ref{thm:rank1_stability} holds under assumption (A1) on the constraint set, and the following assumption (A3) on the distribution of $a_j,b_j$:
\begin{enumerate}
	\item[(A3)] The measurement matrices $\{A_j=a_j b_j^T\}_{j=1}^m$ satisfy that $\{a_j\}_{j=1}^m$ and $\{b_j\}_{j=1}^m$ are independent random vectors, where $\{a_j\}_{j=1}^m$ (resp. $\{b_j\}_{j=1}^m$) are i.i.d. following a distribution $\rmD_1$ (resp. $\rmD_2$) that satisfies the following concentration of measure bounds ($\varepsilon,\delta>0$):
\begin{align*}
&\bbP_{\rmD_1\rmD_2}\left[\norm{a}_2 \leq R_1,~\norm{b}_2 \leq R_2,~\left|a^TXb\right|\leq \delta \right] \leq C_{\rmD_1,\rmD_2,R_1,R_2,\delta}\cdot\frac{\delta}{\varepsilon},\qquad \forall X~~\text{s.t.}~\varepsilon \leq \norm{X}_2 \leq 2, \\
&\bbP_{\rmD_1} [\norm{a}_2 > R_1] = \theta_{\rmD_1,R_1},\qquad \bbP_{\rmD_2} [\norm{b}_2 > R_2] = \theta_{\rmD_2,R_2}.
\end{align*}
The constant $C_{\rmD_1,\rmD_2,R_1,R_2,\delta}$ is independent of $\varepsilon$, but may contain a $\polylog(\delta)$ factor.
\end{enumerate}

\begin{theorem}\label{thm:rank1_stability}
Suppose the constraint set and the measurement matrices satisfy assumptions (A1) and (A3), respectively. If $m>d_i$, $\varepsilon<1$ and $\delta < R_1R_2\rho_i$, then the corresponding stability result in Theorem \ref{thm:stability} holds with probability $1-P_i$, where
\begin{equation*}
P_i \leq C_i\left(3C_{\rmD_1,\rmD_2,R_1,R_2,3\delta}\right)^m\cdot (R_1R_2)^{d_i} \cdot\frac{\delta^{m-d_i}}{\varepsilon^m} + m\left(\theta_{\rmD_1,R_1}+\theta_{\rmD_2,R_2}\right).
\end{equation*}
The norm $\norm{\cdot}_\calX$ in which the recovery error of the matrix is measured in the definition of stability, is the \emph{spectral norm} $\norm{\cdot}_2$.
\end{theorem}

For uniform distributions and i.i.d. Gaussian distributions, the above theorem reduces to Corollaries \ref{cor:rank1_uniform} and \ref{cor:rank1_gaussian}, respectively. The proof of Theorem \ref{thm:rank1_stability} follows steps similar to those in \cite[Lemma 3]{Riegler2015}. Again, our refined covering number argument allows stability with the same number of measurements, and our generalization covers measurement models following other distributions, e.g., Gaussian distribution in Corollary \ref{cor:rank1_gaussian}. For uniformly distributed measurements, our guarantee Corollary \ref{cor:rank1_uniform} and Lemma \ref{lem:concentration2} are analogous to a previous result by Riegler et al. \cite[Lemma 4]{Riegler2015}. Our adaptation of the previous result makes a big difference in our stability guarantees. Please see a detailed comparison in Appendix \ref{app:concentration}.

\begin{corollary}\label{cor:rank1_uniform}
Suppose the constraint set satisfies assumption (A1), the measurement matrices $\{A_j=a_j b_j^T\}_{j=1}^m$ satisfy that $\{a_j\}_{j=1}^m$ and $\{b_j\}_{j=1}^m$ are independent random vectors, where $\{a_j\}_{j=1}^m$ (resp. $\{b_j\}_{j=1}^m$) are i.i.d. following distribution $\rmU_1$ -- uniform distribution on $R_1\calB_{n_1}$ (resp. $\rmU_2$ -- uniform distribution on $R_2\calB_{n_2}$). Then the stability results in Theorem \ref{thm:rank1_stability} hold, except for a small probability:
\begin{equation*}
P_i \leq C_i\left(\frac{12 V_{n_1-1}\cdot V_{n_2-1}}{V_{n_1}\cdot V_{n_2}} \left(1+\ln\frac{2R_1R_2}{3\delta}\right)\right)^m\left(\frac{\delta}{R_1R_2}\right)^{m-d_i} \left(\frac{1}{\varepsilon}\right)^m.
\end{equation*}
\end{corollary}

\begin{corollary}\label{cor:rank1_gaussian}
Suppose the constraint set satisfies assumption (A1), the measurement matrices $\{A_j=a_j b_j^T\}_{j=1}^m$ satisfy that $\{a_j\}_{j=1}^m$ and $\{b_j\}_{j=1}^m$ are independent random vectors, where $\{a_j\}_{j=1}^m$ (resp. $\{b_j\}_{j=1}^m$) are i.i.d. following distribution $\rmG_1$ -- with i.i.d. Gaussian entries $N(0,\sigma_1^2)$ (resp. $\rmG_2$ -- with i.i.d. Gaussian entries $N(0,\sigma_2^2)$). Then the stability results in Theorem \ref{thm:rank1_stability} hold, except for a small probability:
\begin{align*}
P_i \leq & C_i\left(\frac{3}{\sigma_1\sigma_2}\left(1+\ln\left(1+\frac{2\sigma_1\sigma_2}{3\delta}\right)\right)\right)^m\cdot(R_1R_2)^{d_i}\cdot\frac{\delta^{m-d_i}}{\varepsilon^m} \\
& + m e^{-\frac{n_1}{2}\left(\frac{R_1^2}{n_1\sigma_1^2}-1-\ln \frac{R_1^2}{n_1\sigma_1^2} \right)} + m e^{-\frac{n_2}{2}\left(\frac{R_2^2}{n_2\sigma_2^2}-1-\ln \frac{R_2^2}{n_2\sigma_2^2} \right)}, \qquad  \forall R_1>\sqrt{n_1}\sigma_1,~\forall R_2>\sqrt{n_2}\sigma_2.
\end{align*}
\end{corollary}

The sample complexities for stable recovery in Tables \ref{tab:sc} and \ref{tab:sc_alt} also hold for rank-1 measurement matrices. Here, deviation from the true matrix is measured in spectral norm, and the constants in the probability of failure are different from those for stable recovery using unstructured measurement matrices.
\begin{corollary}\label{cor:rank1_mr} 
The stability results in Theorem \ref{thm:rank1_stability}, Corollary \ref{cor:rank1_uniform}, and Corollary \ref{cor:rank1_gaussian} hold
\begin{enumerate}
	\item for $\Omega_\calX$ defined by \eqref{eq:cssp} -- \eqref{eq:cssl}, under the sample complexities in Table \ref{tab:sc}.
	\item for $\Omega_\calX$ defined by \eqref{eq:csl} or \eqref{eq:cssl} when perturbations are small ($\delta<R\rho_i$, $i=1,2,3$), under the less demanding sample complexities in Table \ref{tab:sc_alt}.
\end{enumerate}
\end{corollary}


\subsection{Symmetric Rank-1 Measurement Matrices}\label{sec:symm_rank1}

The analysis in Section \ref{sec:rank1} does not apply to symmetric rank-1 matrices, which arises in applications like phase retrieval. In this section, we treat this case separately. In this section, $n_1=n_2=n$, and $A_j=a_ja_j^T$. Theorem \ref{thm:symm_stability} holds under assumption (A1), and the following assumption (A4) on the distribution of random vectors $\{a_j\}_{j=1}^{m}\subset\bbR^n$. 

\begin{enumerate}
	\item[(A4)] The measurement matrices $\{A_j=a_j a_j^T\}_{j=1}^m$ satisfy that $\{a_j\}_{j=1}^m$ are i.i.d. random vectors following a distribution $\rmD$ that satisfies the following concentration of measure bounds ($\varepsilon,\delta>0$):
\begin{align*}
&\bbP_{\rmD}\left[\norm{a}_2 \leq R,~\left|a^TXa\right|\leq \delta \right] \leq C_{\rmD,R}\cdot\sqrt{\frac{\delta}{\varepsilon}},\qquad \forall X~~\text{s.t.}~\norm{X}_2 \geq \varepsilon, \\
&\bbP_{\rmD} [\norm{a}_2 > R] = \theta_{\rmD,R}.
\end{align*}
The constant $C_{\rmD,R}$ is independent of $\varepsilon,\delta$.
\end{enumerate}

\begin{theorem}\label{thm:symm_stability}
Suppose the constraint set and the measurement matrices satisfy assumptions (A1) and (A3), respectively, and all matrices in $\Omega_\calX$ are symmetric. If $m>2d_i$, $\varepsilon<1$ and $\delta < R^2\rho_i$, then the corresponding stability result in Theorem \ref{thm:stability} holds with probability $1-P_i$, where
\begin{equation*}
P_i \leq C_i\left(\sqrt{3}C_{\rmD,R}\right)^m\cdot R^{2d_i} \cdot\frac{\delta^{m/2-d_i}}{\varepsilon^{m/2}} + m\cdot\theta_{\rmD,R}.
\end{equation*}
The norm $\norm{\cdot}_\calX$ in which the recovery error of the matrix is measured in the definition of stability, is the \emph{spectral norm} $\norm{\cdot}_2$.
\end{theorem}

In phase retrieval, the measurements of a unknown vector $x_0\in\bbR^n$ are obtained without signs. By Theorem \ref{thm:symm_stability}, in the lifted phase retrieval problem, we need $m>2d_1 = 2n$ measurements to stably recover the unknown $n\times n$ symmetric rank-1 matrix $X_0=x_0x_0^T$. By Theorem \ref{thm:stability}, if the measurements are obtained with signs, $m>d_1=n$ measurements are sufficient. Hence, due to the loss of signs, we need twice as many measurements to stably recover the unknown vector.

For uniform distributions and i.i.d. Gaussian distributions, the above theorem reduces to Corollaries \ref{cor:symm_uniform} and \ref{cor:symm_gaussian}, respectively.
\begin{corollary}\label{cor:symm_uniform}
Suppose the constraint set satisfies assumption (A1), the measurement matrices $\{A_j=a_j a_j^T\}_{j=1}^m$ satisfy that $\{a_j\}_{j=1}^m$ are i.i.d. random vectors following distribution $\rmU$ -- uniform distribution on $R\calB_n$. Then the stability results in Theorem \ref{thm:symm_stability} hold, except for a small probability:
\begin{align*}
P_i \leq C_i\left(\frac{2\sqrt{6} \cdot V_{n-1}}{V_n} \right)^m\left(\frac{\delta}{R^2}\right)^{\frac{m}{2}-d_i} \left(\frac{1}{\varepsilon}\right)^\frac{m}{2}.
\end{align*}
\end{corollary}

\begin{corollary}\label{cor:symm_gaussian}
Suppose the constraint set satisfies assumption (A1), the measurement matrices $\{A_j=a_j a_j^T\}_{j=1}^m$ satisfy that $\{a_j\}_{j=1}^m$ are i.i.d. random vectors following distribution $\rmG$ -- with i.i.d. Gaussian entries $N(0,\sigma^2)$. Then the stability results in Theorem \ref{thm:symm_stability} hold, except for a small probability:
\begin{align*}
P_i \leq & C_i\left(\frac{2\sqrt{3}}{\sqrt{\pi}\sigma}\right)^m\cdot R^{2d_i}\cdot\frac{\delta^{m/2-d_i}}{\varepsilon^{m/2}}  + m e^{-\frac{n}{2}\left(\frac{R^2}{n\sigma^2}-1-\ln \frac{R^2}{n\sigma^2} \right)}, \qquad  \forall R>\sqrt{n}\sigma.
\end{align*}
\end{corollary}

Combining the above results with the bounds on covering numbers in Section \ref{sec:cnmd}, we have Corollary \ref{cor:symm_mr}.
\begin{corollary}\label{cor:symm_mr} 
The stability results in Theorem \ref{thm:symm_stability}, Corollary \ref{cor:symm_uniform}, and Corollary \ref{cor:symm_gaussian} hold for set of symmetric matrices $\Omega_\calX$ defined by \eqref{eq:cssp}, \eqref{eq:css}, \eqref{eq:csl2}, or \eqref{eq:cssl2}, under the sample complexities in Table \ref{tab:symm_sc}.
\end{corollary}

\begin{table}[htbp]%
\renewcommand{\arraystretch}{1.3}
\begin{center}
\begin{tabular}{ >{\centering\arraybackslash}m{0.5in} >{\centering\arraybackslash}m{1.5in} >{\centering\arraybackslash}m{1.5in} >{\centering\arraybackslash}m{1.5in} >{\centering\arraybackslash}m{0in} }
\hline
$\Omega_\calX$ & Single point stability on $\Omega_\calB$ & Uniform stability on $\Omega_\calB$ & Uniform stability on $\Omega_\calX$ &\\
\hline
\eqref{eq:cssp} & $m>2t$ & $m>2t$ & $m>2t$ &\\
\eqref{eq:css} & $m>2s$ & $m>4s$ & $m>4s$ &\\
\eqref{eq:csl2} & $m>2nr$ & $m>4nr$ & $m>4nr$ &\\
\eqref{eq:cssl2} & $m>2sr$ & $m>4sr$ & $m>8sr$ &\\
\hline
\end{tabular}
\end{center}
\caption{A summary of sample complexities for stable recovery using symmetric rank-1 measurement matrices.}
\label{tab:symm_sc}
\end{table}


\section{Proofs of Main Results}

\subsection{Proof for Unstructured Measurement Matrices}

\begin{proof}[Proof of Theorem \ref{thm:stability}]
We start with single point stability on $\Omega_\calB$. The measurement matrices $\{A_j\}_{j=1}^{m}$ are i.i.d. random matrices following distribution $\rmD$, which we denote by $\rmD^m$. Then the probability of failure for single point stability is:
\begin{align}
P_1 \defeq &~ 1-\bbP_{\rmD^m}\left[ \forall X\in\Omega_\calB, \text{ if }\norm{\calA(X)-\calA(X_0)}_2\leq \delta, \text{ then } \norm{X-X_0}_\rmF\leq \varepsilon \right] \nonumber\\
= &~ \bbP_{\rmD^m}\left[ \exists X\in\Omega_\calB, \text{ s.t. }\norm{\calA(X)-\calA(X_0)}_2\leq \delta, \text{ and } \norm{X-X_0}_\rmF > \varepsilon \right] \nonumber\\
= &~ \bbP_{\rmD^m}\left[ \exists X\in\Omega_\calB - X_0, \text{ s.t. } \norm{X}_\rmF > \varepsilon \text{ and } \norm{\calA(X)}_2\leq \delta \right] \label{eq:prob_sing_bd}
\end{align}
Define $\Omega_{\varepsilon}\defeq \{X\in\Omega_\calB - X_0: \norm{X}_\rmF>\varepsilon\}$. Then the probability of failure (unstable recovery) is:
\begin{align}
P_1 = & \bbP_{\rmD^m}\left[ \exists X\in\Omega_\varepsilon \text{ s.t. } \norm{\calA(X)}_2\leq \delta \right] \nonumber\\
\leq & \bbP_{\rmD^m}\left[\norm{A_j}_\rmF \leq R,\forall j\in[m],~\text{and}~ \exists X\in\Omega_\varepsilon \text{ s.t. } \norm{\calA(X)}_2\leq \delta \right] + \bbP_{\rmD^m}\left[\exists j\in[m]~\text{s.t.}~\norm{A_j}_\rmF > R \right] \nonumber\\
\leq & \bbP_{\rmD^m}\left[\norm{A_j}_\rmF \leq R,\forall j\in[m],~\text{and}~ \exists X\in\Omega_\varepsilon \text{ s.t. } \norm{\calA(X)}_2\leq \delta \right] + m \cdot \theta_{\rmD,R}, \label{eq:condition_bd}
\end{align}
where \eqref{eq:condition_bd} follows from a union bound and \eqref{eq:def_theta1} in assumption (A2). To complete the proof, we need to bound the first term. 

We form a minimal cover of $\Omega_{\varepsilon}$ with balls of radius $\rho = \frac{\delta}{R}<\rho_1$ centered at the points $\{X_i\}_{i=1}^{N_{\Omega_{\varepsilon}}(\rho)}$. The centers of the balls are not necessarily in $\Omega_{\varepsilon}$. However, by the minimality of the cover, the intersection of $\Omega_{\varepsilon}$ with each ball is nonempty, hence there exists another set of points $\{X'_i\}_{i=1}^{N_{\Omega_{\varepsilon}}(\rho)}$ such that
\[
X'_i \in \Omega_{\varepsilon} \bigcap (X_i+\rho\calB_{n_1\times n_2}), \quad i = 1,2,\cdots,N_{\Omega_\varepsilon}(\rho).
\]
Now we can cover $\Omega_\varepsilon$ with balls of radius $2\rho$ centered at $\{X'_i\}_{i=1}^{N_{\Omega_\varepsilon}(\rho)}$, which are points in $\Omega_\varepsilon$  (a property that will be needed for inequality \eqref{eq:use_A2} below), because
\[
(X_i+\rho\calB_{n_1\times n_2}) \subset (X'_i+2\rho\calB_{n_1\times n_2}), \quad i = 1,2,\cdots,N_{\Omega_\varepsilon}(\rho),
\]
\[
\Omega_\varepsilon \subset \bigcup_{1\leq i\leq N_{\Omega_\varepsilon}(\rho)} (X_i+\rho\calB_{n_1\times n_2}) \subset \bigcup_{1\leq i\leq N_{\Omega_\varepsilon}(\rho)} (X'_i+2\rho\calB_{n_1\times n_2}).
\]

Therefore, the first term in \eqref{eq:condition_bd} satisfies:
\begin{align}
& \bbP_{\rmD^m}\left[\norm{A_j}_\rmF \leq R,\forall j\in[m],~\text{and}~ \exists X\in\Omega_\varepsilon \text{ s.t. } \norm{\calA(X)}_2\leq \delta \right] \nonumber\\
\leq & \sum_{i=1}^{N_{\Omega_\varepsilon}(\rho)} \bbP_{\rmD^m}\left[\norm{A_j}_\rmF \leq R,\forall j\in[m],~\text{and}~\exists X\in (X'_i+2\rho\calB_{n_1\times n_2}), \text{ s.t. } \norm{\calA(X)}_2\leq \delta \right] \label{eq:union} \\
\leq & \sum_{i=1}^{N_{\Omega_\varepsilon}(\rho)} \bbP_{\rmD^m}\left[\norm{A_j}_\rmF \leq R,\forall j\in[m],~\text{and}~ \exists X\in (X'_i+2\rho\calB_{n_1\times n_2}), \text{ s.t. } \left|\left<A_j,X\right>\right| \leq \delta, \forall j\in[m] \right] \label{eq:relax}\\
\leq & \sum_{i=1}^{N_{\Omega_\varepsilon}(\rho)} \bbP_{\rmD^m}\left[\norm{A_j}_\rmF \leq R,~\left|\left<A_j,X'_i\right>\right| \leq 3\delta, \forall j\in[m] \right] \label{eq:triangle}\\
= & \sum_{i=1}^{N_{\Omega_\varepsilon}(\rho)} \Big(\bbP_\rmD\left[\norm{A_1}_\rmF \leq R,~\left|\left<A_1,X'_i\right>\right|\leq 3\delta\right]\Big)^m \label{eq:independent}\\
\leq & N_{\Omega_\varepsilon}(\rho) \left(C_{\rmD,R}\cdot\frac{3\delta}{\varepsilon}\right)^m \label{eq:use_A2}\\
\leq & C_1\left(\frac{R}{\delta}\right)^{d_1} \left(C_{\rmD,R}\cdot\frac{3\delta}{\varepsilon}\right)^m.  \label{eq:use_A1}
\end{align}
Inequality \eqref{eq:union} uses a union bound. The event in \eqref{eq:union} implies the event in \eqref{eq:relax}, which then implies the event in \eqref{eq:triangle}. Inequality \eqref{eq:triangle} is due to the following chain of inequalities, of which the last is implied by $\norm{A_j}_\rmF \leq R$, $\norm{X'_i-X}_\rmF\leq 2\rho$, and $\left|\left<A_j,X\right>\right|\leq \delta$:
\begin{align}
\left|\left<A_j,X'_i\right>\right| \leq & \left|\left<A_j,X'_i-X\right>\right| + \left|\left<A_j,X\right>\right| \nonumber\\
\leq & \norm{A_j}_\rmF\norm{X'_i-X}_\rmF + \left|\left<A_j,X\right>\right|  \nonumber\\
\leq & 2R\rho +\delta ~=~ 3\delta. \label{eq:triangle_detail}
\end{align}
Equation \eqref{eq:independent} is due to the fact that $\{A_j\}_{j=1}^{m}$ are i.i.d. random matrices. Inequality \eqref{eq:use_A2} follows from \eqref{eq:concentration1} in assumption (A2). (By construction, $X'_i$, as points in $\Omega_\varepsilon$, satisfy $\norm{X'_i}_\rmF> \varepsilon$.) Inequality \eqref{eq:use_A1} uses the fact that $N_{\Omega_{\varepsilon}}(\rho) \leq N_{ \Omega_\calB }(\rho)=N_{\Omega_\calB}\left(\frac{\delta}{R}\right)$, and \eqref{eq:cn_bound} in assumption (A1). (By assumption, $\frac{\delta}{R}<\rho_1$.) Replacing the first term in \eqref{eq:condition_bd} by \eqref{eq:use_A1}, we have
\[
P_1 \leq C_1\left(3C_{\rmD,R}\right)^m\cdot R^{d_1} \cdot\frac{\delta^{m-d_1}}{\varepsilon^m} + m\cdot \theta_{\rmD,R},
\]
thus completing the proof of single point stability on $\Omega_\calB$.

Next, we prove uniform stability on $\Omega_\calB$ and $\Omega_\calX$. The probability of failure for uniform stability on bounded constraint set $\Omega_\calB$ is:
\begin{align}
P_2 \defeq& 1-\bbP_{\rmD^m}\left[ \forall X_1,X_2\in\Omega_\calB, \text{ if }\norm{\calA(X_1)-\calA(X_2)}_2\leq \delta, \text{ then } \norm{X_1-X_2}_\rmF\leq \varepsilon \right] \nonumber\\
= & \bbP_{\rmD^m}\left[ \exists X_1,X_2\in\Omega_\calB, \text{ s.t. }\norm{\calA(X_1)-\calA(X_2)}_2\leq \delta, \text{ and } \norm{X_1-X_2}_\rmF > \varepsilon \right] \nonumber\\
= & \bbP_{\rmD^m}\left[ \exists X\in\Omega_{\dB}, \text{ s.t. } \norm{X}_\rmF > \varepsilon \text{ and } \norm{\calA(X)}_2\leq \delta \right]. \label{eq:prob_unif_bd}
\end{align}
The probability of failure for uniform stability on unbounded constraint set $\Omega_\calX$ is:
\begin{align}
P_3 \defeq& 1-\bbP_{\rmD^m}\left[ \forall X_1,X_2\in\Omega_\calX, \text{ if }\norm{\calA(X_1)-\calA(X_2)}_2\leq \delta, \text{ then } \norm{X_1-X_2}_\rmF\leq \varepsilon \right] \nonumber\\
= & \bbP_{\rmD^m}\left[ \exists X_1,X_2\in\Omega_\calX, \text{ s.t. }\norm{\calA(X_1)-\calA(X_2)}_2\leq \delta, \text{ and } \norm{X_1-X_2}_\rmF > \varepsilon \right] \nonumber\\
= & \bbP_{\rmD^m}\left[ \exists X\in\Omega_\calX-\Omega_\calX, \text{ s.t. } \norm{X}_\rmF > \varepsilon \text{ and } \norm{\calA(X)}_2\leq \delta \right] \label{eq:prob_unif_ubd0}\\
= & \bbP_{\rmD^m}\left[ \exists X\in\Omega_{\dX}, \text{ s.t. } \norm{X}_\rmF > \varepsilon \text{ and } \norm{\calA(X)}_2\leq \delta \right]. \label{eq:prob_unif_ubd}
\end{align}
The last line owes to the fact that the events in \eqref{eq:prob_unif_ubd0} and \eqref{eq:prob_unif_ubd}, which we denote by $E_1$ and $E_2$, are equivalent for the following reason: First, $\Omega_{\dX}= (\Omega_\calX-\Omega_\calX)\bigcap \calB_{n_1\times n_2} \subset \Omega_\calX-\Omega_\calX$, hence $E_2$ implies $E_1$. Secondly, suppose $E_1$ is true, i.e., there exists $X\in \Omega_\calX-\Omega_\calX$ such that $\norm{X}_\rmF > \varepsilon$ and $\norm{\calA(X)}_2\leq \delta$. If $\norm{X}_\rmF\leq 1$, then $X\in\Omega_{\dX}$ and $E_2$ is true. If $\norm{X}_\rmF> 1$, then $\norm{\frac{X}{\norm{X}_\rmF}}_\rmF = 1>\varepsilon$ and $\norm{\calA(\frac{X}{\norm{X}_\rmF})}_2\leq \frac{\delta}{\norm{X}_\rmF}<\delta$, hence $\frac{X}{\norm{X}_\rmF}\in\Omega_{\dX}$ and $E_2$ is true. In either case, $E_1$ implies $E_2$. Therefore, $E_1$ and $E_2$ are equivalent.

We continue the proof of uniform stability on $\Omega_\calB$ and $\Omega_\calX$. Comparing \eqref{eq:prob_unif_bd} and \eqref{eq:prob_unif_ubd} to \eqref{eq:prob_sing_bd}, we argue that the rest of the proof of single point stability on $\Omega_\calB$ applies, with $\Omega_{\dB}$ and $\Omega_{\dX}$ replacing $\Omega_\calB-X_0$. Therefore, with $d_1,C_1$ replaced by $d_2,C_2$ (resp. $d_3,C_3$) in the sample complexity and the probability of stable recovery, the single point stability result translates to uniform stability, thus completing the proof of uniform stability on $\Omega_\calB$ (resp. $\Omega_\calX$).
\end{proof}

\begin{proof}[Proof of Corollary \ref{cor:uniform}]
Corollary \ref{cor:uniform} follows Theorem \ref{thm:stability}, with the following expressions for $\theta_{\rmU,R}$ and $C_{\rmU,R}$ for uniform distribution on the ball $R\calB_{n_1n_2}$:
\[
\theta_{\rmU,R} = 0,
\]
\[
C_{\rmU,R} = \frac{2 V_{n_1n_2-1}}{R\cdot V_{n_1n_2}}.
\]
The expression for $C_{\rmU,R}$ follows from Lemma \ref{lem:concentration1} in Appendix \ref{app:concentration}. If $\norm{X}_\rmF\geq \varepsilon$, then
\[
\bbP_\rmU\left[\norm{A}_\rmF \leq R,~\left|\left<A,X\right>\right|\leq \delta\right] = \bbP_\rmU\left[\left|\left<A,X\right>\right|\leq \delta\right] \leq \frac{2\delta \cdot V_{n_1n_2-1}}{\varepsilon R\cdot V_{n_1n_2}},
\]
thus we have the expression for $C_{\rmU,R}$.
\end{proof}

\begin{proof}[Proof of Corollary \ref{cor:gaussian}]
Corollary \ref{cor:gaussian} follows Theorem \ref{thm:stability}, with the following expressions for $\theta_{\rmG,R}$ and $C_{\rmG,R}$ for i.i.d. Gaussian distribution $\rmG$:
\begin{align*}
\theta_{\rmG,R} = &~ \bbP_\rmG \left[\norm{A}_\mathrm{F}^2 > R^2\right] \leq e^{-\frac{n_1n_2}{2}\left(\frac{R^2}{n_1n_2\sigma^2}-1-\ln \frac{R^2}{n_1n_2\sigma^2} \right)},\qquad\text{if}~R^2>n_1n_2\sigma^2,
\end{align*}
\[
C_{\rmG,R} = \frac{\sqrt{2}}{\sqrt{\pi}\sigma}.
\]
The expression for $\theta_{\rmG,R}$ follows from Chernoff bound.\footnote{This probability is small. For example, if $R=2\sqrt{n_1n_2}\sigma$, then $\theta_{\rmG,R}\leq e^{-0.8n_1n_2}$.} 
The expression for $C_{\rmG,R}$ follows from Lemma \ref{lem:concentration1_gaussian} in Appendix \ref{app:concentration}. If $\norm{X}_\rmF\geq \varepsilon$, then
\begin{align*}
\bbP_\rmG\left[\norm{A}_\rmF \leq R,~\left|\left<A,X\right>\right|\leq \delta\right] \leq \bbP_\rmG\left[\left|\left<A,X\right>\right|\leq \delta\right] \leq \frac{\sqrt{2}\delta}{\sqrt{\pi}\sigma\varepsilon},
\end{align*}
thus we have the expression for $C_{\rmG,R}$.
\end{proof}

\subsection{Proof for Rank-1 Measurement Matrices}
\begin{proof}[Proof of Theorem \ref{thm:rank1_stability}]
The proof follows steps mostly analogous to those in the proof of Theorem \ref{thm:stability} , with the Frobenius norm replaced by the spectral norm.

Define $\Omega_{\varepsilon}\defeq \{X\in\Omega_\calB - X_0: \norm{X}_2>\varepsilon\}$. Then \eqref{eq:condition_bd} is replaced by:
\begin{align}
P_1 = &~ \bbP_{\rmD_1^m\rmD_2^m}\left[ \exists X\in\Omega_\varepsilon \text{ s.t. } \norm{\calA(X)}_2\leq \delta \right] \nonumber\\
\leq &~ \bbP_{\rmD_1^m\rmD_2^m}\left[\norm{a_j}_2\leq R_1,\norm{b_j}_2 \leq R_2,~\forall j\in[m],~\text{and}~ \exists X\in\Omega_\varepsilon \text{ s.t. } \norm{\calA(X)}_2\leq \delta \right]\nonumber\\
&~ + \bbP_{\rmD_1^m\rmD_2^m}\left[\exists j\in[m]~\text{s.t.}~\norm{a_j}_2 > R_1 ~\text{or}~\norm{b_j}_2 \leq R_2\right] \nonumber\\
\leq &~ \bbP_{\rmD_1^m\rmD_2^m}\left[\norm{a_j}_2\leq R_1,\norm{b_j}_2 \leq R_2,~\forall j\in[m],~\text{and}~ \exists X\in\Omega_\varepsilon \text{ s.t. } \norm{\calA(X)}_2\leq \delta \right]+ m (\theta_{\rmD_1,R_1}+\theta_{\rmD_2,R_2}). \label{eq:condition_bd2}
\end{align}
Let $\rho = \frac{\delta}{R_1R_2}<\rho_1$. To bound the first term, we find points $\{X'_i\}_{i=1}^{N_{\Omega_{\varepsilon}}(\rho)}$ such that
\[
\Omega_{\varepsilon} \subset  \bigcup_{1\leq i\leq N_{\Omega_\varepsilon}(\rho)} (X'_i+2\rho\calB_{n_1\times n_2}).
\]
Then, the first term in \eqref{eq:condition_bd2} satisfies
\begin{align}
& \bbP_{\rmD_1^m\rmD_2^m}\left[\norm{a_j}_2\leq R_1,\norm{b_j}_2 \leq R_2,~\forall j\in[m],~\text{and}~ \exists X\in\Omega_\varepsilon \text{ s.t. } \norm{\calA(X)}_2\leq \delta \right] \nonumber\\
\leq & \sum_{i=1}^{N_{\Omega_\varepsilon}(\rho)} \Bigl(\bbP_{\rmD_1\rmD_2}\left[\norm{a_1}_2\leq R_1,~\norm{b_1}_2\leq R_2,~\left|a_1^TX'_ib_1 \right| \leq 3\delta \right]\Bigr)^m \label{eq:triangle2} \\
\leq & C_1\left(\frac{R_1R_2}{\delta}\right)^{d_1}\left(C_{\rmD_1,\rmD_2,R_1,R_2,3\delta}\cdot\frac{3\delta}{\varepsilon}\right)^m. \label{eq:use_A3}
\end{align}
Inequality \eqref{eq:triangle2} uses the following chain of inequalities:
\begin{align*}
\left|a_j^TX'_ib_j\right| \leq & \left|a_j^T(X'_i-X)b_j\right| + \left|a_j^TXb_j\right| \\
\leq & \norm{a_j}_2\norm{X'_i-X}_2\norm{b_j}_2 + \left|a_j^TXb_j\right|  \\
\leq & \norm{a_j}_2\norm{X'_i-X}_\rmF\norm{b_j}_2 + \left|a_j^TXb_j\right|  \\
\leq & 2R_1R_2\rho +\delta ~=~ 3\delta.
\end{align*}
Inequality \eqref{eq:use_A3} follows from assumptions (A1) and (A3). (By construction, $X'_i$, as points in $\Omega_\varepsilon$, satisfy $\varepsilon < \norm{X'_i}_2 \leq 2$.) To complete the proof for single point stability, we substitute \eqref{eq:use_A3} into \eqref{eq:condition_bd2}.

Uniform stability on $\Omega_\calB$ and $\Omega_\calX$, using rank-1 measurement matrices, can be proved by replacing $\Omega_\calB-X_0$ with $\Omega_{\dB}$ and $\Omega_{\dX}$, respectively.
\end{proof}

\begin{proof}[Proof of Corollary \ref{cor:rank1_uniform}]
Corollary \ref{cor:rank1_uniform} follows Theorem \ref{thm:rank1_stability}, with the following expressions for $\theta_{\rmU_1,R_1}$, $\theta_{\rmU_2,R_2}$ and $C_{\rmU_1,\rmU_2,R_1,R_2,\delta}$:
\[
\theta_{\rmU_1,R_1}=\theta_{\rmU_2,R_2}=0,
\]
\[
C_{\rmU_1,\rmU_2,R_1,R_2,\delta} = \frac{4 V_{n_1-1}\cdot V_{n_2-1}}{R_1R_2\cdot V_{n_1}\cdot V_{n_2}} \left(1+\ln\frac{2R_1R_2}{\delta}\right).
\]
The expression for $C_{\rmU_1,\rmU_2,R_1,R_2,\delta}$ follows from Lemma \ref{lem:concentration2}. If $\varepsilon\leq \norm{X}_2\leq 2$, then
\begin{align}
& \bbP_{\rmU_1\rmU_2}\left[\norm{a}_2\leq R_1,~\norm{b}_2\leq R_2,~\left|a^TXb \right| \leq \delta \right] \nonumber\\
= & \bbP_{\rmU_1\rmU_2}\left[\left|a^TXb \right| \leq \delta\right] \nonumber\\
\leq & \frac{4\delta \cdot V_{n_1-1}\cdot V_{n_2-1}}{\varepsilon R_1R_2\cdot V_{n_1}\cdot V_{n_2}} \left(1+\ln\frac{2R_1R_2}{\delta}\right), \nonumber
\end{align}
thus we have the expression for $C_{\rmU_1,\rmU_2,R_1,R_2,\delta}$.
\end{proof}

\begin{proof}[Proof of Corollary \ref{cor:rank1_gaussian}]
Corollary \ref{cor:rank1_gaussian} follows Theorem \ref{thm:rank1_stability}, with the following expressions for $\theta_{\rmG_1,R_1}$, $\theta_{\rmG_2,R_2}$ and $C_{\rmG_1,\rmG_2,R_1,R_2,\delta}$:
\[
\text{For}~i=1,2,\quad \theta_{\rmG_i,R_i} = \bbP_{\rmG_i} \left[\norm{a}_2^2 > R_i^2\right] \leq e^{-\frac{n_i}{2}\left(\frac{R_i^2}{n_i\sigma_i^2}-1-\ln \frac{R_i^2}{n_i\sigma_i^2} \right)},\qquad\text{if}~R_i^2>n_i\sigma_i^2,
\]
\[
C_{\rmG_1,\rmG_2,R_1,R_2,\delta} = \frac{1}{\sigma_1\sigma_2}\left(1+\ln\left(1+\frac{2\sigma_1\sigma_2}{\delta}\right)\right).
\]
The expressions for $\theta_{\rmG_i,R_i}$ follows from Chernoff bound. The expression for $C_{\rmG_1,\rmG_2,R_1,R_2,\delta}$ follows from Lemma \ref{lem:concentration2_gaussian}. If $\varepsilon\leq \norm{X}_2\leq 2$, then
\begin{align*}
\bbP_{\rmG_1\rmG_2}\left[\norm{a}_2\leq R_1,~\norm{b}_2\leq R_2,~\left|a^TXb \right| \leq\delta \right] \leq \bbP_{\rmG_1\rmG_2}\left[\left|a^TXb \right| \leq \delta \right] \leq  \frac{\delta}{\varepsilon\sigma_1\sigma_2}\left(1+\ln\left(1+\frac{2\sigma_1\sigma_2}{\delta}\right)\right),
\end{align*}
thus we have the expression for $C_{\rmG_1,\rmG_2,R_1,R_2,\delta}$. 
\end{proof}


\subsection{Proof for Symmetric Rank-1 Measurement Matrices}
\begin{proof}[Proof of Theorem \ref{thm:symm_stability}]
The proof follows steps mostly analogous to those in the proofs of Theorems \ref{thm:stability} and \ref{thm:rank1_stability}.

Define $\Omega_{\varepsilon}\defeq \{X\in\Omega_\calB - X_0: \norm{X}_2>\varepsilon\}$. Then \eqref{eq:condition_bd} is replaced by:
\begin{align}
P_1 \leq &~ \bbP_{\rmD^m}\left[\norm{a_j}_2\leq R,~\forall j\in[m],~\text{and}~ \exists X\in\Omega_\varepsilon \text{ s.t. } \norm{\calA(X)}_2\leq \delta \right]+ m \cdot \theta_{\rmD,R}. \label{eq:condition_bd3}
\end{align}
Let $\rho = \frac{\delta}{R^2}<\rho_1$. To bound the first term, we find points $\{X'_i\}_{i=1}^{N_{\Omega_{\varepsilon}}(\rho)}$ such that
\[
\Omega_{\varepsilon} \subset  \bigcup_{1\leq i\leq N_{\Omega_\varepsilon}(\rho)} (X'_i+2\rho\calB_{n\times n}).
\]
Then, the first term in \eqref{eq:condition_bd3} satisfies:
\begin{align}
& \bbP_{\rmD^m}\left[\norm{a_j}_2\leq R,~\forall j\in[m],~\text{and}~ \exists X\in\Omega_\varepsilon \text{ s.t. } \norm{\calA(X)}_2\leq \delta \right] \nonumber\\
\leq & \sum_{i=1}^{N_{\Omega_\varepsilon}(\rho)} \Bigl(\bbP_{\rmD}\left[\norm{a_1}_2\leq R,~\left|a_1^TX'_ia_1 \right| \leq 3\delta \right]\Bigr)^m \label{eq:triangle3} \\
\leq & C_1\left(\frac{R^2}{\delta}\right)^{d_1}\left(C_{\rmD,R}\cdot\sqrt{\frac{3\delta}{\varepsilon}}\right)^m. \label{eq:use_A4}
\end{align}
Inequality \eqref{eq:triangle3} uses the following chain of inequalities:
\begin{align*}
\left|a_j^TX'_ia_j\right| \leq & \left|a_j^T(X'_i-X)a_j\right| + \left|a_j^TXa_j\right| \\
\leq & \norm{a_j}_2\norm{X'_i-X}_2\norm{a_j}_2 + \left|a_j^TXa_j\right|  \\
\leq & \norm{a_j}_2\norm{X'_i-X}_\rmF\norm{a_j}_2 + \left|a_j^TXa_j\right|  \\
\leq & 2R^2\rho +\delta ~=~ 3\delta. 
\end{align*}
Inequality \eqref{eq:use_A4} follows from assumptions (A1) and (A4). (By construction, $X'_i$, as points in $\Omega_\varepsilon$, satisfy $\norm{X'_i}_2 > \varepsilon$.) To complete the proof for single point stability, we substitute \eqref{eq:use_A4} into \eqref{eq:condition_bd3}.

Uniform stability on $\Omega_\calB$ and $\Omega_\calX$, using symmetric rank-1 measurement matrices, can be proved by replacing $\Omega_\calB-X_0$ with $\Omega_{\dB}$ and $\Omega_{\dX}$, respectively.
\end{proof}

\begin{proof}[Proof of Corollary \ref{cor:symm_uniform}]
Corollary \ref{cor:symm_uniform} follows Theorem \ref{thm:symm_stability}, with the following expressions for $\theta_{\rmU,R}$ and $C_{\rmU,R}$:
\[
\theta_{\rmU,R} = 0,
\]
\[
C_{\rmU,R} = \frac{2\sqrt{2} V_{n-1}}{ R\cdot V_n}.
\]
The expression for $C_{\rmU,R}$ follows from Lemma \ref{lem:concentration3}. If $\norm{X}_2\geq \varepsilon$, then
\begin{align}
& \bbP_{\rmU}\left[\norm{a}_2\leq R,~\left|a^TXa \right| \leq \delta \right] \nonumber\\
\leq & \bbP_{\rmU}\left[\left|a^TXa\right| \leq \delta\right] \nonumber\\
\leq & \frac{2\sqrt{2\delta} \cdot V_{n-1}}{\sqrt{\varepsilon} R\cdot V_n},
\end{align}
thus we have the expression for $C_{\rmU,R}$.
\end{proof}

\begin{proof}[Proof of Corollary \ref{cor:symm_gaussian}]
Corollary \ref{cor:symm_gaussian} follows Theorem \ref{thm:symm_stability}, with the following expressions for $\theta_{\rmG,R}$ and $C_{\rmG,R}$:
\[
\theta_{\rmG,R} = \bbP_{\rmG} \left[\norm{a}_2^2 > R^2\right] \leq e^{-\frac{n}{2}\left(\frac{R^2}{n\sigma^2}-1-\ln \frac{R^2}{n\sigma^2} \right)},\qquad\text{if}~R^2>n\sigma^2,
\]
\[
C_{\rmG,R} = \frac{2}{\sqrt{\pi}\sigma}.
\]
The expression for $\theta_{\rmG,R}$ follows from Chernoff bound. The expression for $C_{\rmG,R}$ follows from Lemma \ref{lem:concentration3_gaussian}. If $\norm{X}_2\geq \varepsilon$, then
\begin{align*}
\bbP_{\rmG}\left[\norm{a}_2\leq R,~\left|a^TXa \right| \leq \delta \right] \leq \bbP_{\rmG}\left[\left|a^TXa\right| \leq \delta \right] \leq  \frac{2\sqrt{\delta}}{\sqrt{\pi\varepsilon}\sigma},
\end{align*}
thus we have the expression for $C_{\rmG,R}$.
\end{proof}

\section{Nonlinearity in Matrix Recovery}
\subsection{Parameterized Constraint Set}
We addressed single point stability and uniform stability on a bounded constraint set $\Omega_\calB$ in Section \ref{sec:main} as part of the main results. The bounded constraint sets we considered so far all satisfy $\Omega_\calB=\Omega_\calX \bigcap \calB_{n_1\times n_2}$, where $\Omega_\calX$ is a cone \eqref{eq:cssp} -- \eqref{eq:cssl2}. However, stability on a bounded constraint set only requires that $\Omega_\calB$ has a bound on its covering number, $N_{\Omega_\calB}(\rho)\leq C_1(\frac{1}{\rho})^{d_1}$, where $d_1$ is an upper bound on the number of degrees of freedom. In this section, we show that this type of covering number bound can derived for a much larger class of bounded constraint sets with a small number of degrees of freedom, and hence the stability results also apply to these constraint sets.

In \eqref{eq:cssp} and \eqref{eq:css}, the constraint set is a subspace and a union of subspaces, respectively. 

\subsection{Nonlinear Transform of Measurement Matrices}

\subsection{Nonlinear Measurement Operator}

\section{Conclusions} \label{sec:conclusions}
We studied the optimal sample complexity of the matrix recovery problem. If the measurement matrices follow distributions specified in this paper, then under optimal sample complexities, the recovery is stable with high probability against small perturbations in the measurements.

\appendix

\section{Proofs of Covering Number Bounds}\label{app:cn}

\subsection{Useful Lemmas about Covering Numbers}
In this appendix, we prove the bounds on covering numbers in Section \ref{sec:cnmd}. We start with two lemmas, which will be used later.
\begin{lemma}\label{lem:product}
Let $\Omega_\calU$ and $\Omega_\calV$ be nonempty subsets of $\sigma\calB_{n_1\times r}$ and $\sigma\calB_{n_2\times r}$, respectively. If $\Omega_\calB \subset \{UV^T\in \bbR^{n_1\times n_2}: U\in\Omega_\calU, V\in\Omega_\calV\}$, then $N_{\Omega_\calB}(\rho)\leq N_{\Omega_\calU}(\frac{\rho}{2\sigma})N_{\Omega_\calV}(\frac{\rho}{2\sigma})$.
\end{lemma}

\begin{proof}
We cover $\Omega_\calU \subset \sigma\calB_{n_1\times r}$ and $\Omega_\calV \subset \sigma\calB_{n_2\times r}$ with balls of radius $\frac{\rho}{2\sigma}$ centered at the following two sets of points, respectively:
\[
\{U_i\}_{i=1}^{N_{\Omega_\calU}(\frac{\rho}{2\sigma})} \subset \sigma\calB_{n_1\times r},\qquad  \{V_i\}_{i=1}^{N_{\Omega_\calV}(\frac{\rho}{2\sigma})} \subset \sigma\calB_{n_2\times r}.
\]
Since $\Omega_\calB \subset \{UV^T\in \bbR^{n_1\times n_2}: U\in\Omega_\calU, V\in\Omega_\calV\}$, any $X\in\Omega_\calB$ can be written as $X=UV^T$ for some $U\in \Omega_\calU$ and $V\in\Omega_\calV$. Then we can find centers of the above coverings, $U_{i_1}$ and $V_{i_2}$, such that
\[
\norm{U-U_{i_1}}_\rmF \leq \frac{\rho}{2\sigma},\qquad
\norm{V-V_{i_2}}_\rmF \leq \frac{\rho}{2\sigma}.
\]
Then
\begin{align*}
\norm{X-U_{i_1}V_{i_2}^T}_\rmF =& \norm{UV^T-U_{i_1}V^T+U_{i_1}V^T-U_{i_1}V_{i_2}^T}_\rmF\\
\leq & \norm{U-U_{i_1}}_2\norm{V}_\rmF+\norm{V-V_{i_2}}_2\norm{U_{i_1}}_\rmF\\
\leq & \norm{U-U_{i_1}}_\rmF\norm{V}_\rmF+\norm{V-V_{i_2}}_\rmF\norm{U_{i_1}}_\rmF\\
\leq & \frac{\rho}{2\sigma}\times \sigma\times 2 = \rho.
\end{align*}
Therefore, the set $\Omega_\calB$ can be covered by $N_{\Omega_\calU}(\frac{\rho}{2\sigma})N_{\Omega_\calV}(\frac{\rho}{2\sigma})$ balls in $\bbR^{n_1\times n_2}$ of radius $\rho$, centered at the matrices (like $U_{i_1}V_{i_2}^T$) generated by the centers of the coverings of $\Omega_\calU$ and $\Omega_\calV$. It follows that
\[
N_{\Omega_\calB}(\rho)\leq N_{\Omega_\calU}(\frac{\rho}{2\sigma})N_{\Omega_\calV}(\frac{\rho}{2\sigma}).
\]
\end{proof}

\begin{lemma}\label{lem:symm_product}
Let $\Omega_\calU$ be a nonempty subset of $\sigma\calB_{n\times r}$. If $\Lambda$ is a diagonal matrix whose entries are $\pm 1$, and $\Omega_\calB \subset \{U\Lambda U^T\in \bbR^{n\times n}: U\in\Omega_\calU\}$, then $N_{\Omega_\calB}(\rho)\leq N_{\Omega_\calU}(\frac{\rho}{2\sigma})$.
\end{lemma}

\begin{proof}
We cover $\Omega_\calU \subset \sigma\calB_{n\times r}$ with balls of radius $\frac{\rho}{2\sigma}$ centered at the following set of points:
\[
\{U_i\}_{i=1}^{N_{\Omega_\calU}(\frac{\rho}{2\sigma})} \subset \sigma\calB_{n\times r}.
\]
Since $\Omega_\calB \subset \{U\Lambda U^T\in \bbR^{n\times n}: U\in\Omega_\calU\}$, any $X\in\Omega_\calB$ can be written as $X=U\Lambda U^T$ for some $U\in \Omega_\calU$. Then we can find a center of the above covering, $U_i$, such that
\[
\norm{U-U_i}_\rmF \leq \frac{\rho}{2\sigma}.
\]
Then
\begin{align*}
\norm{X-U_i\Lambda U_i^T}_\rmF =& \norm{U\Lambda U^T-U_i\Lambda U^T+U_i\Lambda U^T-U_i\Lambda U_i^T}_\rmF\\
\leq & \norm{U-U_i}_2\norm{\Lambda}_2\norm{U}_\rmF+\norm{U_i}_2\norm{\Lambda}_2\norm{U-U_i}_\rmF\\
\leq & \norm{U-U_i}_\rmF\norm{U}_\rmF+\norm{U_i}_\rmF\norm{U-U_i}_\rmF\\
\leq & \frac{\rho}{2\sigma}\times \sigma\times 2 = \rho.
\end{align*}
Therefore, the set $\Omega_\calB$ can be covered by $N_{\Omega_\calU}(\frac{\rho}{2\sigma})$ balls in $\bbR^{n_n\times n_n}$ of radius $\rho$, centered at the matrices (like $U_i\Lambda U_i^T$) generated by the centers of the coverings of $\Omega_\calU$. It follows that
\[
N_{\Omega_\calB}(\rho)\leq N_{\Omega_\calU}(\frac{\rho}{2\sigma}).
\]
\end{proof}

\begin{lemma}\label{lem:difference}
Let $\Omega_\calB$ be a nonempty bounded subset of $\bbR^{n_1\times n_2}$. Then $N_{\Omega_\calB-\Omega_\calB}(\rho)\leq N_{\Omega_\calB}(\frac{\rho}{2})^2$, and $\overline{\dim}_\mathrm{B}(\Omega_\calB-\Omega_\calB)\leq 2\overline{\dim}_\mathrm{B}(\Omega_\calB)$.
\end{lemma}

\begin{proof}
We cover $\Omega_\calB$ with balls of radius $\frac{\rho}{2}$ centered at $\{X_i\}_{i=1}^{N_{\Omega_\calB}(\frac{\rho}{2})}$. Any point $X_1-X_2\in\Omega_\calB-\Omega_\calB$ is the difference of two points in $\Omega_\calB$. Then we can find centers of the above covering, $X_{i_1}$ and $X_{i_2}$, such that
\[
\norm{X_1-X_{i_1}}_\rmF \leq \frac{\rho}{2},\qquad
\norm{X_2-X_{i_2}}_\rmF \leq \frac{\rho}{2}.
\]
Then
\begin{align*}
\norm{(X_1-X_2)-(X_{i_1}-X_{i_2})}_\rmF =& \norm{(X_1-X_{i_1})-(X_2-X_{i_2})}_\rmF\\
\leq & \norm{X_1-X_{i_1}}_\rmF+\norm{X_2-X_{i_2}}_\rmF\\
\leq & \frac{\rho}{2}\times 2 = \rho.
\end{align*}
Therefore, the set $\Omega_\calB-\Omega_\calB$ can be covered by $N_{\Omega_\calB}(\frac{\rho}{2})^2$ balls in $\bbR^{n_1\times n_2}$ of radius $\rho$, centered at the matrices (like $X_{i_1}-X_{i_2}$) generated by the centers of the coverings of $\Omega_\calB$. It follows that
\[
N_{\Omega_\calB-\Omega_\calB}(\rho)\leq N_{\Omega_\calB}(\frac{\rho}{2})^2.
\]
Therefore,
\[
\overline{\dim}_\mathrm{B}(\Omega_\calB-\Omega_\calB) = \underset{\rho\rightarrow 0}{\lim\sup}\frac{\log N_{\Omega_\calB-\Omega_\calB}(\rho)}{\log\frac{1}{\rho}} \leq \underset{\rho\rightarrow 0}{\lim\sup}\frac{2\log N_{\Omega_\calB}(\frac{\rho}{2})}{\log\frac{1}{\rho}} = 2\overline{\dim}_\mathrm{B}(\Omega_\calB).
\]
\end{proof}

\subsection{Proof of Proposition \ref{pro:cnB}}

Next, we prove Proposition \ref{pro:cnB}. We split the proof into six parts, bounding the covering numbers of different $\Omega_\calB$'s corresponding to different $\Omega_\calX$'s defined by \eqref{eq:cssp} -- \eqref{eq:cssl2}.

\begin{proof}[Proof of Proposition \ref{pro:cnB} ($\Omega_\calX$ defined by \eqref{eq:cssp})]
Since $\{M_i\}_{i=1}^{t}$ is an orthonormal basis, we have $\norm{\sum_{i=1}^{t}\beta^{(i)} M_i}_\rmF=\norm{\beta}_2$. Hence
\[
\Omega_\calB = \Omega_\calX \bigcap \calB_{n_1\times n_2} = \{X \in\bbR^{n_1\times n_2}:\exists \beta\in\calB_{t},~\text{s.t.}~X = \sum_{i=1}^{t}\beta^{(i)} M_i\}.
\]
We cover $\calB_t$ with balls of radius $\rho$ centered at the points $\{\beta_j\}_{j=1}^{N_{\calB_t}(\rho)}$. Then for every $X = \sum_{i=1}^{t}\beta^{(i)} M_i\in\Omega_\calB$, there exists a center $\beta_j$ such that $\norm{\beta-\beta_j}\leq \rho$, and hence
\[
\norm{X-\sum_{i=1}^{t}\beta_j^{(i)} M_i}_\rmF = \norm{\sum_{i=1}^{t}(\beta^{(i)}-\beta_j^{(i)}) M_i}_\rmF = \norm{\beta-\beta_j}_2 \leq \rho.
\]
Then we can cover $\Omega_\calB$ with $N_{\calB_t}(\rho)$ balls of radius $\rho$, centered at points $\left\{\sum_{i=1}^{t}\beta_j^{(i)} M_i\right\}_{j=1}^{N_{\calB_t}(\rho)}$. Therefore, for $0<\rho< 1$,
\[
N_{\Omega_\calB}(\rho)\leq N_{\calB_t}(\rho)\leq \left(\frac{3}{\rho}\right)^t.
\]
The covering number of the ball $\calB_t$ follows from a standard volume argument \cite{Pollard1990}.
\end{proof}

\begin{proof}[Proof of Proposition \ref{pro:cnB} ($\Omega_\calX$ defined by \eqref{eq:css})]
The set $\Omega_\calB=\Omega_\calX\bigcap \calB_{n_1\times n_2}$ is
\begin{align*}
\Omega_\calB = & \{X\in\bbR^{n_1\times n_2}: \exists \beta\in\bbR^t,~\text{s.t.}~\norm{\beta}_0\leq s,~X = \sum_{i=1}^{t}\beta^{(i)} M_i, ~ \norm{X}_\rmF \leq 1\}\\
\subset & \{X\in\bbR^{n_1\times n_2}: \exists \beta\in\bbR^t,~\text{s.t.}~\norm{\beta}_0\leq s, ~ \norm{\beta}_2 \leq \frac{1}{\sigma_{s,\min}},~X = \sum_{i=1}^{t}\beta^{(i)} M_i\}\\
= & \{X\in\bbR^{n_1\times n_2}: \exists \beta\in \frac{1}{\sigma_{s,\min}}\calB_t,~\text{s.t.}~\norm{\beta}_0\leq s, ~X = \sum_{i=1}^{t}\beta^{(i)} M_i\}.
\end{align*}
Define set $\Omega_\beta \defeq \{\beta\in \frac{1}{\sigma_{s,\min}}\calB_t:~\norm{\beta}_0\leq s\}$. If we cover $\Omega_\beta$ with balls of radius $\frac{\rho}{\sigma_{2s,\max}}$, centered at $\{\beta_j\}_{j=1}^{\rho/\sigma_{2s,\max}}\subset \Omega_\beta$, then for every $X=\sum_{i=1}^{t}\beta^{(i)} M_i\in\Omega_\calB$, there exists a center $\beta_j$ in the above cover that satisfies $\norm{\beta-\beta_j}_2\leq \frac{\rho}{\sigma_{2s,\max}}$, and hence
\[
\norm{X-\sum_{i=1}^{t}\beta_j^{(i)} M_i}_\rmF = \norm{\sum_{i=1}^{t}(\beta^{(i)}-\beta_j^{(i)}) M_i}_\rmF \leq \sigma_{2s,\max} \norm{\beta-\beta_j}_2 \leq \rho.
\]
Therefore, the covering number of $\Omega_\calB$ satisfies:
\[
N_{\Omega_\calB}\leq N_{\Omega_\beta}\left(\frac{\rho}{\sigma_{2s,\max}}\right) \leq {t \choose s}\left(3\cdot \frac{1}{\sigma_{s,\min}}\cdot\frac{\sigma_{2s,\max}}{\rho}\right)^s\leq {t \choose s}\left(\frac{3\kappa_{2s}}{\rho}\right)^s.
\]
The second inequality is due to the fact that $\Omega_\beta$ is the union of ${t \choose s}$ balls in subspaces of dimension $s$, of radius $\frac{1}{\sigma_{s,\min}}$. The third inequality follows from $\sigma_{s,\min}\geq \sigma_{2s,\min}$.
\end{proof}

\begin{proof}[Proof of Proposition \ref{pro:cnB} ($\Omega_\calX$ defined by \eqref{eq:csl})]
For an arbitrary low-rank matrix $X$ in the unit ball,
\[
X\in\Omega_\calB=\{X\in\bbR^{n_1\times n_2}: \norm{X}_\rmF\leq 1, \rank(X)\leq r\},
\]
the singular value decomposition is $X=U\Sigma V^T=(U\Sigma^\frac{1}{2})(V\Sigma^\frac{1}{2})^T$, where $U\in\bbR^{n_1\times r}$ and $V\in\bbR^{n_2\times r}$ have orthonormal columns, and $\Sigma = \diag([\sigma_1,\sigma_2,\cdots,\sigma_r])$. The Frobenius norm of $U\Sigma^\frac{1}{2}$ and $V\Sigma^\frac{1}{2}$ satisfies:
\begin{align*}
\norm{U\Sigma^\frac{1}{2}}_\rmF = \norm{V\Sigma^\frac{1}{2}}_\rmF = \norm{\Sigma^\frac{1}{2}}_\rmF = & \sqrt{\sigma_1+\sigma_2+\cdots+\sigma_r} \leq \sqrt{\sqrt{\sigma_1^2+\sigma_2^2+\cdots+\sigma_r^2}\sqrt{r}} \leq r^\frac{1}{4},
\end{align*}
where the first inequality follows from the Cauchy–Schwarz inequality, and the second inequality is due to the fact $\norm{X}_\rmF=\sqrt{\sigma_1^2+\cdots+\sigma_r^2}\leq 1$.
Therefore,
\[
\Omega_\calB \subset \{UV^T\in\bbR^{n_1\times n_2}: U\in\Omega_\calU, V\in\Omega_\calV\},
\]
where
\[
\Omega_\calU = \{U\in\bbR^{n_1\times r}: \norm{U}_\rmF\leq r^\frac{1}{4}\},\qquad \Omega_\calV = \{V\in\bbR^{n_2\times r}: \norm{V}_\rmF\leq r^\frac{1}{4}\}.
\]
By a standard volume argument:
\[
N_{\Omega_\calU}(\rho) \leq \left(\frac{3r^\frac{1}{4}}{\rho}\right)^{n_1r},\qquad N_{\Omega_\calV}(\rho) \leq \left(\frac{3r^\frac{1}{4}}{\rho}\right)^{n_2r}.
\]
It follows from Lemma \ref{lem:product} that
\begin{align*}
N_{\Omega_\calB}(\rho) \leq N_{\Omega_\calU}(\frac{\rho}{2r^\frac{1}{4}})N_{\Omega_\calV}(\frac{\rho}{2r^\frac{1}{4}})\leq \left(\frac{6\sqrt{r}}{\rho}\right)^{(n_1+n_2)r}.
\end{align*}
\end{proof}

\begin{proof}[Proof of Proposition \ref{pro:cnB} ($\Omega_\calX$ defined by \eqref{eq:cssl})]
The proof is analogous to the previous case, with $\Omega_\calU$ and $\Omega_\calV$ replaced by:
\[
\Omega_\calU = \{U\in\bbR^{n_1\times r}: \norm{U}_{\mathrm{r},0}\leq s_1,\norm{U}_\rmF\leq r^\frac{1}{4}\},\qquad \Omega_\calV = \{V\in\bbR^{n_2\times r}: \norm{V}_{\mathrm{r},0}\leq s_2,\norm{V}_\rmF\leq r^\frac{1}{4}\},
\]
\[
N_{\Omega_\calU}(\rho) \leq {n_1\choose s_1}\left(\frac{3r^\frac{1}{4}}{\rho}\right)^{s_1r},\qquad N_{\Omega_\calV}(\rho) \leq {n_2\choose s_2}\left(\frac{3r^\frac{1}{4}}{\rho}\right)^{s_2r}.
\]
Therefore,
\begin{align*}
N_{\Omega_\calB}(\rho) \leq N_{\Omega_\calU}(\frac{\rho}{2r^\frac{1}{4}})N_{\Omega_\calV}(\frac{\rho}{2r^\frac{1}{4}}) \leq {n_1\choose s_1}{n_2\choose s_2}\left(\frac{6\sqrt{r}}{\rho}\right)^{(s_1+s_2)r}.
\end{align*}
\end{proof}

\begin{proof}[Proof of Proposition \ref{pro:cnB} ($\Omega_\calX$ defined by \eqref{eq:csl2})]
For an arbitrary low-rank symmetric matrix $X$ in the unit ball,
\[
X\in\Omega_\calB=\{X\in\bbR^{n\times n}: \norm{X}_\rmF\leq 1, X=X^T, \rank(X)\leq r\},
\]
the eigendecomposition is $X=U\Lambda U^T$, where $U\in\bbR^{n\times r}$ has orthonormal columns, and $\Lambda = \diag([\lambda_1,\lambda_2,\cdots,\lambda_r])$. The eigenvalues $\lambda_1\geq \lambda_2\geq \lambda_r$ can be positive, zero, or negative. Define
\begin{align*}
\Lambda_+\defeq & \diag([|\lambda_1|,|\lambda_2|,\cdots,|\lambda_r|]),\\
\Lambda_{\sgn} \defeq & \diag([\ind{\lambda_1\geq 0},\ind{\lambda_2\geq 0},\cdots,\ind{\lambda_r\geq 0}]).
\end{align*}
Then $X = (U\Lambda_+^\frac{1}{2})\Lambda_{\sgn}(U\Lambda_+^\frac{1}{2})^T$. By an argument analogous to that in the proof of case \eqref{eq:csl}:
\begin{align*}
\norm{U\Lambda_+^\frac{1}{2}}_\rmF = \norm{\Lambda_+^\frac{1}{2}}_\rmF = & \sqrt{|\lambda_1|+|\lambda_3|+\cdots+|\lambda_r|} \leq \sqrt{\sqrt{\lambda_1^2+\lambda_2^2+\cdots+\lambda_r^2}\sqrt{r}} \leq r^\frac{1}{4}.
\end{align*}
Therefore,
\[
\Omega_\calB \subset \bigcup\limits_{j=0,1,\cdots,r}\{U\Lambda_j U^T\in\bbR^{n\times n}: U\in\Omega_\calU\},
\]
where $\Lambda_j$ is a diagonal matrix, whose first $j$ diagonal entries are $1$, and last $(r-j)$ diagonal entries are $-1$, and
\[
\qquad\Omega_\calU = \{U\in\bbR^{n\times r}: \norm{U}_\rmF\leq r^\frac{1}{4}\},\qquad N_{\Omega_\calU}(\rho) \leq \left(\frac{3r^\frac{1}{4}}{\rho}\right)^{nr}.
\]
It follows from Lemma \ref{lem:symm_product} that
\begin{align*}
N_{\Omega_\calB}(\rho) \leq (r+1)N_{\Omega_\calU}(\frac{\rho}{2r^\frac{1}{4}}) \leq (r+1)\left(\frac{6\sqrt{r}}{\rho}\right)^{nr}.
\end{align*}
\end{proof}

\begin{proof}[Proof of Proposition \ref{pro:cnB} ($\Omega_\calX$ defined by \eqref{eq:cssl2})]
The proof is analogous to the previous case, with $\Omega_\calU$ replaced by:
\[
\Omega_\calU = \{U\in\bbR^{n\times r}: \norm{U}_{\mathrm{r},0}\leq s,\norm{U}_\rmF\leq r^\frac{1}{4}\},\qquad N_{\Omega_\calU}(\rho) \leq {n\choose s}\left(\frac{3r^\frac{1}{4}}{\rho}\right)^{sr}.
\]
Therefore,
\begin{align*}
N_{\Omega_\calB}(\rho) \leq (r+1)N_{\Omega_\calU}(\frac{\rho}{2r^\frac{1}{4}}) \leq (r+1){n\choose s}\left(\frac{6\sqrt{r}}{\rho}\right)^{sr}.
\end{align*}
\end{proof}

\subsection{Proof of Propositions \ref{pro:cndB} and \ref{pro:cndX}}

Next, we prove Propositions \ref{pro:cndB} and \ref{pro:cndX}. Using the bounds on the covering number of $\Omega_\calB$ in Proposition \ref{pro:cnB}, it is easy to acquire bounds on the covering numbers of $\Omega_{\dB}$ and $\Omega_{\dX}$.

\begin{proof}[Proof of Proposition \ref{pro:cndB}]
When $\Omega_\calX$ is defined by \eqref{eq:cssp},
\begin{align*}
\Omega_{\dB} = & \{X \in\bbR^{n_1\times n_2}:\exists \beta\in\calB_t-\calB_t,~\text{s.t.}~X = \sum_{i=1}^{t}\beta^{(i)} M_i\} \\
= & \{X \in\bbR^{n_1\times n_2}:\exists \beta\in 2\calB_t,~\text{s.t.}~X = \sum_{i=1}^{t}\beta^{(i)} M_i\}.
\end{align*}
By the proof of Proposition \ref{pro:cnB}, when $\Omega_\calX$ is defined by \eqref{eq:cssp}, we have
\[
N_{\Omega_{\dB}}(\rho)\leq N_{2\calB_t}(\rho)\leq \left(\frac{6}{\rho}\right)^t.
\]

When $\Omega_\calX$ is defined by \eqref{eq:css} -- \eqref{eq:cssl2}, we apply Lemma \ref{lem:difference} to $\Omega_{\dB} = \Omega_\calB-\Omega_\calB$. 
If the covering number of $\Omega_\calB$ satisfies $N_{\Omega_\calB}(\rho) \leq C_1\left(\frac{1}{\rho}\right)^{d_1}$, where $C_1$ is independent of $\rho$, then the covering number of $\Omega_{\dB}$ satisfies $N_{\Omega_{\dB}} (\rho) \leq C_1^2\left(\frac{2}{\rho}\right)^{2d_1}$. Let $d_2 = 2d_1$ and $C_2 =2^{2d_1} C_1^2 $. Then, the rest of the bounds in Proposition \ref{pro:cndB} follow from their counterparts in Proposition \ref{pro:cnB}.
\end{proof}

\begin{proof}[Proof of Proposition \ref{pro:cndX}]
When $\Omega_\calX$ is the subspace defined by \eqref{eq:cssp}, $\Omega_\calX-\Omega_\calX=\Omega_\calX$. Hence
\[
\Omega_{\dX} = (\Omega_\calX-\Omega_\calX)\bigcap \calB_{n_1\times n_2} = \Omega_\calX \bigcap \calB_{n_1\times n_2} = \Omega_\calB.
\]
Therefore, when $\Omega_\calX$ is defined by \eqref{eq:cssp}, we have
\[
N_{\Omega_{\dX}}(\rho)\leq \left(\frac{3}{\rho}\right)^t.
\]

When $\Omega_\calX$ is defined by \eqref{eq:css} -- \eqref{eq:cssl2}, we use the fact that the sparsity (resp. rank) of matrices in $\Omega_\calX-\Omega_\calX$ is bounded by twice the sparsity (resp. rank) of matrices in $\Omega_\calX$. Therefore, the rest of the bounds in Proposition \ref{pro:cndX} follow from their counterparts in Proposition \ref{pro:cnB}, with $2s$, $2r$, $2s_1$, and $2s_2$ replacing $s$, $r$, $s_1$, and $s_2$. 
\end{proof}

\subsection{Proof of Alternative Bounds Using Minkowski Dimensions}
Next we prove Propositions \ref{pro:cnmd}.

\begin{proof}[Proof of Proposition \ref{pro:cnmd}]
If
\[
\overline{\dim}_\mathrm{B}(\Omega_\calB) =\underset{\rho\rightarrow 0}{\lim\sup}\frac{\log N_{\Omega_\calB}(\rho)}{\log\frac{1}{\rho}} \leq d,
\]
then, by the definition of limit superior, there exists $\rho_0>0$ such that for all $0<\rho<\rho_0$,
\[
\frac{\log N_{\Omega_\calB}(\rho)}{\log\frac{1}{\rho}} \leq d+1,
\]
i.e.,
\[
N_{\Omega_\calB}(\rho) \leq \left(\frac{1}{\rho}\right)^{d+1}.
\]
\end{proof}

Corollary \ref{cor:cn_alt} follows from Proposition \ref{pro:cnmd} and the Minkowski dimension bounds on $\Omega_\calB$, $\Omega_{\dB}$, and $\Omega_{\dX}$. We give the bound on the Minkowski dimension of $\Omega_\calB$ in the following lemma. Then the Minkowski dimension of $\Omega_{\dB}=\Omega_\calB-\Omega_\calB$ can be bounded using Lemma \ref{lem:difference} : $\overline{\dim}_\mathrm{B}(\Omega_{\dB}) \leq 2\overline{\dim}_\mathrm{B}(\Omega_\calB)$. The Minkowski dimension of $\Omega_{\dX}=(\Omega_\calX-\Omega_\calX)\bigcap \calB_{n_1\times n_2}$ has the same bound as $\Omega_\calB=\Omega_\calX \bigcap \calB_{n_1\times n_2}$, with $2r$, $2s_1$, and $2s_2$ replacing $r$, $s_1$, and $s_2$.

\begin{lemma}\label{lem:md}
The upper Minkowski dimension of $\Omega_\calB$ has the following bound:
\begin{enumerate}
	\item $\overline{\dim}_\mathrm{B}(\Omega_\calB)\leq (n_1+n_2-r)r$, if $\Omega_\calX$ is the set of low-rank matrices in \eqref{eq:csl}.
	\item $\overline{\dim}_\mathrm{B}(\Omega_\calB)\leq (s_1+s_2-r)r$, if $\Omega_\calX$ is the set of sparse low-rank matrices in \eqref{eq:cssl}.
\end{enumerate}
\end{lemma}

\begin{proof}[Proof of Lemma \ref{lem:md}]
The first Minkowski dimension bound is given by Lemma 1 in \cite{Riegler2015}.

We prove the second Minkowski dimension bound using an argument similar to that of \cite[Lemma 1]{Riegler2015}. We first rewrite $\Omega_\calX$ in \eqref{eq:cssl} as the following union of subsets:
\[\def\arraystretch{0.6}
\Omega_\calX = \underset{\begin{array}{c}
\scriptstyle k \leq r \\
\scriptstyle J_1\subset[n_1],k \leq |J_1| = \ell_1 \leq s_1 \\
\scriptstyle J_2\subset[n_2],k \leq |J_2| = \ell_2 \leq s_2 
\end{array}}{\bigcup} \Omega_{k,J_1,J_2},
\]
where
\[
\Omega_{k,J_1,J_2} = \{X \in\bbR^{n_1\times n_2}: \rank(X)=k, X^{(J_1^c,:)}=0, X^{(:,J_2^c)}=0\}
\]
is an embedded submanifold of $\bbR^{n_1\times n_2}$ of dimension $(\ell_1+\ell_2-k)k$ (see \cite[Example 5.30]{Lee2001}). 

By Properties (i) and (ii) in \cite[Section 3.2]{Falconer1990}, the upper Minkowski dimension of $\Omega_{k,J_1,J_2}\bigcap\calB_{n_1\times n_2}$ is bounded by the dimension of the smooth submanifold $\Omega_{k,J_1,J_2}$, which is $(\ell_1+\ell_2-k)k$. 
By Property (iii) in \cite[Section 3.2]{Falconer1990}, the Minkowski dimension is finitely stable, i.e., the Minkowski dimension of a finite union of sets is no more than the sum of the Minkowski dimensions of these sets. Since
\[\def\arraystretch{0.6}
\Omega_\calB = \Omega_\calX\bigcap \calB_{n_1\times n_2} = \underset{\begin{array}{c}
\scriptstyle k \leq r \\
\scriptstyle J_1\subset[n_1],k \leq |J_1| = \ell_1 \leq s_1 \\
\scriptstyle J_2\subset[n_2],k \leq |J_2| = \ell_2 \leq s_2 
\end{array}}{\bigcup} \left( \Omega_{k,J_1,J_2} \bigcap \calB_{n_1\times n_2} \right),
\]
we have
\begin{align*}
\overline{\dim}_\mathrm{B}(\Omega_\calB) \leq & \underset{\def\arraystretch{0.6}\begin{array}{c}
\scriptstyle k \leq r \\
\scriptstyle J_1\subset[n_1],k \leq |J_1| = \ell_1 \leq s_1 \\
\scriptstyle J_2\subset[n_2],k \leq |J_2| = \ell_2 \leq s_2 
\end{array}}{\max} \overline{\dim}_\mathrm{B}\left(\Omega_{k,J_1,J_2} \bigcap \calB_{n_1\times n_2}\right) \\
= & \underset{\scriptstyle k \leq r,~ k \leq \ell_1 \leq s_1,~ k \leq \ell_2 \leq s_2}{\max} (\ell_1+\ell_2-k)k \\
= & (s_1+s_2-r)r.
\end{align*}
\end{proof}

\section{Proof of Concentration of Measure Inequalities}\label{app:concentration}
\begin{lemma}\label{lem:concentration1}
Suppose $A\in\bbR^{n_1\times n_2}$ is a random matrix following a uniform distribution on $R\calB_{n_1\times n_2}$. If a matrix $X\in\bbR^{n_1\times n_2}$ satisfies $\norm{X}_\rmF\geq \varepsilon$, then
\[
\bbP_\rmU\left[\left|\left<A,X\right>\right|\leq \delta \right] \leq \frac{2\delta\cdot V_{n_1n_2-1}}{\varepsilon R\cdot V_{n_1n_2}}.
\]
\end{lemma}

\begin{proof}\footnote{We would like to acknowledge that Lemma \ref{lem:concentration1} is inspired by, and slightly tighter than, \cite[Lemma 3]{Stotz2013}.}
Let $a = \vect(A)\in\bbR^{n_1n_2}$, and $x=\vect(X)\in\bbR^{n_1n_2}$. Then $a$ is a random vector following a uniform distribution on $R\calB_{n_1n_2}$, and $x$ satisfies $\norm{x}_2\geq\varepsilon$. It follows that
\begin{align}
\bbP_\rmU\left[\left|\left<A,X\right>\right|\leq \delta \right] = \bbP_\rmU\left[\left|a^Tx\right|\leq \delta \right] = \bbP_\rmU\left[\left|a^T\frac{x}{\norm{x}_2}\right|\leq \frac{\delta}{\norm{x}_2} \right] = \bbP_\rmU\left[\left|a^Te_1\right|\leq \frac{\delta}{\norm{x}_2} \right], \label{eq:isotropy}
\end{align}
where $e_1$ denotes the first standard basis vector in $\bbR^{n_1n_2}$, $e^{(1)}=1$, $e^{(2:n_1n_2)}=0$, and the last equality follows from the isotropy of $\rmU$.

Therefore,
\begin{align}
\bbP_\rmU\left[\left|\left<A,X\right>\right|\leq \delta \right] = & \bbP_\rmU\left[\left|a^Te_1\right|\leq \frac{\delta}{\norm{x}_2} \right]  \nonumber\\
= & \frac{\int_{R\calB_{n_1n_2}} ~\rmd a ~\ind{\left|a^{(1)}\right|\leq \frac{\delta}{\norm{x}_2} } }{\int_{R\calB_{n_1n_2}} ~\rmd a } \nonumber\\
=& \frac{1}{R^{n_1n_2}V_{n_1n_2}} \int_{R\calB_{n_1n_2-1}}  ~\rmd a^{(2:n_1n_2)} ~ \int_{|a^{(1)}|^2 \leq R^2-\norm{a^{(2:n_1n_2)}}_2^2} ~\rmd a^{(1)} ~\ind{|a^{(1)}|\leq \frac{\delta}{\norm{x}_2} } \nonumber\\
\leq & \frac{1}{R^{n_1n_2}V_{n_1n_2}} \int_{R\calB_{n_1n_2-1}}  ~\rmd a^{(2:n_1n_2)} ~ \int_{-R}^R ~\rmd a^{(1)} ~\ind{|a^{(1)}|\leq \frac{\delta}{\norm{x}_2}}  \nonumber\\
= & \frac{R^{n_1n_2-1}V_{n_1n_2-1}}{R^{n_1n_2}V_{n_1n_2}} ~ \int_{-R}^R ~\rmd a^{(1)} ~\ind{|a^{(1)}|\leq \frac{\delta}{\norm{x}_2}} \nonumber \\
\leq & \frac{V_{n_1n_2-1}}{R\cdot V_{n_1n_2}} ~ \frac{2\delta}{\norm{x}_2} \nonumber\\
\leq & \frac{2\delta\cdot V_{n_1n_2-1}}{\varepsilon R\cdot V_{n_1n_2}}. \nonumber
\end{align}
\end{proof}

\begin{lemma}\label{lem:concentration1_gaussian}
Suppose $A\in\bbR^{n_1\times n_2}$ is a random matrix, whose entries are i.i.d. following a Gaussian distribution $N(0,\sigma^2)$. If a matrix $X\in\bbR^{n_1\times n_2}$ satisfies $\norm{X}_\rmF\geq \varepsilon$, then
\[
\bbP_\rmG\left[\left|\left<A,X\right>\right|\leq \delta \right] \leq \frac{\sqrt{2}\delta}{\sqrt{\pi}\sigma\varepsilon}.
\]
\end{lemma}
\begin{proof}
Since i.i.d. Gaussian distribution is also isotropic, we have \eqref{eq:isotropy} with $\rmG$ replacing $\rmU$:
\begin{align}
\bbP_\rmG\left[\left|\left<A,X\right>\right|\leq \delta \right] = \bbP_\rmG\left[\left|a^Te_1\right|\leq \frac{\delta}{\norm{x}_2} \right] = \bbP_\rmG\left[\left|a^{(1)}\right|\leq \frac{\delta}{\norm{x}_2} \right]. \label{eq:isotropy_gaussian}
\end{align}
Since the entries of $a$ are independent, the probability in \eqref{eq:isotropy_gaussian} only has to do with the marginal distribution of its first entry $\rmG^{(1)}$, which is $N(0,\sigma^2)$, on the interval $\left[-\frac{\delta}{\norm{x}_2},\frac{\delta}{\norm{x}_2}\right]$. Therefore,
\begin{align}
\bbP_\rmG\left[\left|\left<A,X\right>\right|\leq \delta \right] =\bbP_{\rmG^{(1)}}\left[\left|a^{(1)}\right|\leq \frac{\delta}{\norm{x}_2} \right]\leq  p_{\rmG^{(1)}}(0) \cdot \frac{2\delta}{\norm{x}_2} \leq \frac{2\delta}{\sqrt{2\pi}\sigma \varepsilon}.
\end{align}
\end{proof}

\begin{lemma}\label{lem:concentration2}
Suppose $a\in\bbR^{n_1}$ and $b\in\bbR^{n_2}$ are independent random vectors, following uniform distributions on $R_1\calB_{n_1}$ and $R_2\calB_{n_2}$, respectively. If a matrix $X\in\bbR^{n_1\times n_2}$ satisfies $\varepsilon\leq \norm{X}_2\leq E$, then
\[
\bbP_{\rmU_1\rmU_2}\left[\left|a^TX b\right|\leq \delta \right] \leq \frac{4\delta \cdot V_{n_1-1}\cdot V_{n_2-1}}{\varepsilon R_1R_2\cdot V_{n_1}\cdot V_{n_2}} \left(1+\ln\frac{ER_1R_2}{\delta}\right).
\]
\end{lemma}
\begin{proof}\footnote{Lemma \ref{lem:concentration2} is a rephrase of \cite[Lemma A.1]{Li2015d}. We include the proof here for completeness.}
Suppose the singular value decomposition (SVD) of $X$ is $X = U\Sigma V^T$, where $U\in\bbR^{n_1\times n_1}$ and $V\in\bbR^{n_2\times n_2}$ are orthogonal matrices, and $\Sigma\in\bbR^{n_1\times n_2}$ satisfies $\varepsilon\leq \Sigma^{(1,1)} = \norm{X}_2\leq E$. 

Let $\tilde{a} \defeq U^Ta$, and $\tilde{b} \defeq V^Tb$, then $\tilde{a}$ and $\tilde{b}$ are also independent random vectors, following uniform distributions on $R_1\calB_{n_1}$ and $R_2\calB_{n_2}$, respectively. Therefore,
\begin{align}
& \bbP_{\rmU_1\rmU_2}\left[\left|a^TXb\right|\leq \delta\right] \nonumber\\
= & \bbP_{\rmU_1\rmU_2}\left[\left|\tilde{a}^T\Sigma\tilde{b}\right|\leq \delta\right] \nonumber\\
= & \frac{\int_{R_1\calB_{n_1}} ~\rmd\tilde{a}\int_{R_2\calB_{n_2}}~ \rmd\tilde{b} ~\ind{|\tilde{a}^T\Sigma\tilde{b}|\leq \delta} }{\int_{R_1\calB_{n_1}} ~\rmd\tilde{a}\int_{R_2\calB_{n_2}}~\rmd\tilde{b} } \nonumber\\
=& \frac{1}{R_1^{n_1}V_{n_1}\cdot R_2^{n_2} V_{n_2}} \int\limits_{R_1\calB_{n_1-1}}  ~\rmd\tilde{a}^{(2:n_1)} \int\limits_{R_2\calB_{n_2-1}}~\rmd\tilde{b}^{(2:n_2)}~ \phi(\tilde{a},\tilde{b}),  \label{eq:big}
\end{align}
where
\begin{align}
\phi(\tilde{a},\tilde{b}) =& \int_{-R_1}^{R_1} ~\rmd\tilde{a}^{(1)} \int_{-R_2}^{R_2} ~\rmd\tilde{b}^{(1)}~\ind{|\tilde{a}^T\Sigma\tilde{b}|\leq \delta} \cdot \ind{|\tilde{a}^{(1)}|^2 \leq R_1^2-\norm{\tilde{a}^{(2:n_1)}}_2^2} \cdot \ind{|\tilde{b}^{(1)}|^2 \leq R_2^2-\norm{\tilde{b}^{(2:n_2)}}_2^2} \nonumber\\
\leq & \int_{-R_1}^{R_1} ~\rmd\tilde{a}^{(1)} \int_{-R_2}^{R_2} ~\rmd\tilde{b}^{(1)} \ind{\left|\tilde{b}^{(1)}+\frac{1}{\norm{X}_2 \tilde{a}^{(1)}}\tilde{a}^{(2:n_1)T}\Sigma^{(2:n_1,2:n_2)}\tilde{b}^{(2:n_2)}\right| \leq \frac{\delta}{\norm{X}_2|\tilde{a}^{(1)}|}}             \nonumber\\
\leq & \int_{-R_1}^{R_1} ~\rmd\tilde{a}^{(1)} \min\Biggl(\frac{2\delta}{\norm{X}_2|\tilde{a}^{(1)}|}, 2R_2 \Biggl) \nonumber \\
=  & \frac{4\delta}{\norm{X}_2} \left(1+\ln\frac{\norm{X}_2R_1R_2}{\delta}\right) \nonumber \\
\leq & \frac{4\delta}{\varepsilon}  \left(1+\ln\frac{ER_1R_2}{\delta}\right). \label{eq:small}
\end{align}
Substituting \eqref{eq:small} into \eqref{eq:big}, we obtain
\[
\bbP_{\rmU_1\rmU_2}\left[\left|a^TXb\right|\leq \delta\right] \leq \frac{4\delta \cdot R_1^{n_1-1}V_{n_1-1}\cdot R_2^{n_2-1}V_{n_2-1}}{\varepsilon\cdot R_1^{n_1}V_{n_1}\cdot R_2^{n_2}V_{n_2}} \left(1+\ln\frac{ER_1R_2}{\delta}\right) = \frac{4\delta \cdot V_{n_1-1}\cdot V_{n_2-1}}{\varepsilon R_1R_2\cdot V_{n_1}\cdot V_{n_2}} \left(1+\ln\frac{ER_1R_2}{\delta}\right).
\]
\end{proof}

Lemma \ref{lem:concentration2} adapts a previous result by Riegler et al. \cite[Lemma 4]{Riegler2015}. They have two concentration bounds, for $X$ of rank $1$ and for $X$ of rank larger than $1$. Their bound for $\rank(X)> 1$ is tighter in terms of dependence on $\delta$, but is also inversely proportional to the product of all nonzero singular values of $X$. When those singular values decay fast, this bound is not necessarily stronger than our bound. In the analysis of stability, these concentration bounds must apply to an adversarial $X$. The improvement of the dependence of such bounds on $\delta$ is not necessary, and the worse dependence on $X$ becomes problematic. Therefore, our adaptation of the previous result makes a big difference in our stability guarantees.

\begin{lemma}\label{lem:concentration2_gaussian}
Suppose $a\in\bbR^{n_1}$ and $b\in\bbR^{n_2}$ are independent random vectors, and the entries of $a$ (resp. $b$) are i.i.d. following a Gaussian distribution $N(0,\sigma_1^2)$ (resp. $N(0,\sigma_2^2)$). If a matrix $X\in\bbR^{n_1\times n_2}$ satisfies $\varepsilon\leq \norm{X}_2\leq E$, then
\[
\bbP_{\rmG_1\rmG_2}\left[\left|a^TX b\right|\leq \delta \right] \leq \frac{\delta}{\varepsilon\sigma_1\sigma_2}\left(1+\ln\left(1+\frac{E\sigma_1\sigma_2}{\delta}\right)\right).
\]
\end{lemma}
\begin{proof}
Similar to the proof of Lemma \ref{lem:concentration2}, we use the SVD $X=U\Sigma V^T$, and the change of variables $\tilde{a}=U^Ta$, $\tilde{b}=V^Tb$. Since i.i.d. Gaussian distributions are isotropic, $\tilde{a}$ and $\tilde{b}$ follow distributions $\rmG_1$ and $\rmG_2$, respectively, the same distributions as $a$ and $b$. Therefore,
\begin{align}
& \bbP_{\rmG_1\rmG_2}\left[\left|a^TXb\right|\leq \delta\right] \nonumber\\
= & \bbP_{\rmG_1\rmG_2}\left[\left|\tilde{a}^T\Sigma\tilde{b}\right|\leq \delta\right] \nonumber\\
= & \int_{\bbR^{n_1}}~\rmd\tilde{a} \int_{\bbR^{n_2}}~ \rmd\tilde{b} ~\ind{|\tilde{a}^T\Sigma\tilde{b}|\leq \delta} \cdot p_{\rmG_2}(\tilde{b})\cdot p_{\rmG_1}(\tilde{a}) \nonumber \\
= & \int_{\bbR^{n_1-1}}~\rmd\tilde{a}^{(2:n_1)} \int_{\bbR^{n_2-1}}~ \rmd\tilde{b}^{(2:n_2)} ~\phi(\tilde{a},\tilde{b})\cdot p_{\rmG_2^{(2:n_2)}}(\tilde{b}^{(2:n_2)})\cdot p_{\rmG_1^{(2:n_1)}}(\tilde{a}^{(2:n_1)}), \label{eq:big2}
\end{align}
where
\begin{align}
\phi(\tilde{a},\tilde{b}) = & \int_{-\infty}^{\infty}~\rmd\tilde{a}^{(1)} \int_{-\infty}^{\infty}~ \rmd\tilde{b}^{(1)} ~\ind{|\tilde{a}^T\Sigma\tilde{b}|\leq \delta} \cdot p_{\rmG_2^{(1)}}(\tilde{b}^{(1)}) \cdot p_{\rmG_1^{(1)}}(\tilde{a}^{(1)}) \nonumber\\
= & \frac{1}{2\pi \sigma_1\sigma_2}\int_{-\infty}^{\infty}~\rmd\tilde{a}^{(1)} \int_{-\infty}^{\infty}~ \rmd\tilde{b}^{(1)} ~\ind{\left|\tilde{b}^{(1)}+\frac{1}{\norm{X}_2 \tilde{a}^{(1)}}\tilde{a}^{(2:n_1)T}\Sigma^{(2:n_1,2:n_2)}\tilde{b}^{(2:n_2)}\right| \leq \frac{\delta}{\norm{X}_2|\tilde{a}^{(1)}|}}  \nonumber\\
& ~~~~\cdot e^{-\frac{(\tilde{a}^{(1)})^2}{2\sigma_1^2}-\frac{(\tilde{b}^{(1)})^2}{2\sigma_2^2}} \nonumber\\
\leq & \frac{1}{2\pi \sigma_1\sigma_2}\int_{-\infty}^{\infty}~\rmd\tilde{a}^{(1)} \int_{-\infty}^{\infty}~ \rmd\tilde{b}^{(1)} ~\ind{\left|\tilde{b}^{(1)}\right| \leq \frac{\delta}{\norm{X}_2|\tilde{a}^{(1)}|}} \cdot e^{-\frac{(\tilde{a}^{(1)})^2}{2\sigma_1^2}-\frac{(\tilde{b}^{(1)})^2}{2\sigma_2^2}} \nonumber \\
= & \frac{1}{2\pi \sigma_1\sigma_2} \int_{-\infty}^{\infty}~\rmd\tilde{a}^{(1)} \int_{-\infty}^{\infty}~ \rmd\tilde{b}^{(1)} ~\ind{\left|\tilde{a}^{(1)}\tilde{b}^{(1)}\right| \leq \frac{\delta}{\norm{X}_2}} \cdot e^{-\frac{(\tilde{a}^{(1)})^2}{2\sigma_1^2}-\frac{(\tilde{b}^{(1)})^2}{2\sigma_2^2}} \nonumber\\
= & \frac{1}{2\pi} \int_{-\infty}^{\infty}~\rmd u \int_{-\infty}^{\infty}~ \rmd v ~\ind{\left|uv\right| \leq \frac{\delta}{\norm{X}_2\sigma_1\sigma_2}} \cdot e^{-\frac{u^2+v^2}{2}}, \nonumber
\end{align}
where $u = \frac{\tilde{a}^{(1)}}{\sigma_1}$, and $v=\frac{\tilde{b}^{(1)}}{\sigma_2}$. Rewrite the integral in polar coordinates:
\begin{align}
\phi(\tilde{a},\tilde{b}) \leq & \frac{1}{2\pi} \int_{0}^{\infty} \int_{0}^{2\pi}~\ind{\left|\frac{r^2}{2}\sin 2\theta\right| \leq \frac{\delta}{\norm{X}_2\sigma_1\sigma_2}} ~\rmd \theta ~ e^{-\frac{r^2}{2}} \cdot r ~\rmd r \nonumber\\
= & \frac{1}{2\pi} \int_{0}^{\sqrt{\frac{2\delta}{\norm{X}_2\sigma_1\sigma_2}}} \int_{0}^{2\pi} ~\rmd \theta ~ e^{-\frac{r^2}{2}} \cdot r ~\rmd r \nonumber\\
& + \frac{1}{2\pi} \int_{\sqrt{\frac{2\delta}{\norm{X}_2\sigma_1\sigma_2}}}^{\infty} \int_{0}^{2\pi}~\ind{\left|\sin 2\theta\right| \leq \frac{2\delta}{r^2\norm{X}_2\sigma_1\sigma_2}} ~\rmd \theta ~ e^{-\frac{r^2}{2}} \cdot r ~\rmd r \nonumber\\
= & 1-e^{-\frac{\delta}{\norm{X}_2\sigma_1\sigma_2}} \nonumber\\
& + \frac{4}{2\pi} \int_{\sqrt{\frac{2\delta}{\norm{X}_2\sigma_1\sigma_2}}}^{\infty} \int_{0}^{2\pi}~\ind{\left|2\theta\right| \leq \arcsin\frac{2\delta}{r^2\norm{X}_2\sigma_1\sigma_2}} ~\rmd \theta ~ e^{-\frac{r^2}{2}} \cdot r ~\rmd r \nonumber\\
\leq & 1-e^{-\frac{\delta}{\norm{X}_2\sigma_1\sigma_2}} \nonumber\\
& + \frac{4}{2\pi} \int_{\sqrt{\frac{2\delta}{\norm{X}_2\sigma_1\sigma_2}}}^{\infty} \int_{0}^{2\pi}~\ind{\left|2\theta\right| \leq \frac{\pi\delta}{r^2\norm{X}_2\sigma_1\sigma_2}} ~\rmd \theta ~ e^{-\frac{r^2}{2}} \cdot r ~\rmd r \nonumber\\
\leq & 1-e^{-\frac{\delta}{\norm{X}_2\sigma_1\sigma_2}} + \frac{\delta}{\norm{X}_2\sigma_1\sigma_2} \int_{\frac{\delta}{\norm{X}_2\sigma_1\sigma_2}}^{\infty} \frac{1}{z} ~ e^{-z}  ~\rmd z \nonumber\\
\leq & \frac{\delta}{\norm{X}_2\sigma_1\sigma_2} + \frac{\delta}{\norm{X}_2\sigma_1\sigma_2}\ln\left(1+\frac{\norm{X}_2\sigma_1\sigma_2}{\delta}\right) \label{eq:use_exp}\\
\leq & \frac{\delta}{\varepsilon\sigma_1\sigma_2}\left(1+\ln\left(1+\frac{E\sigma_1\sigma_2}{\delta}\right)\right), \label{eq:small2}
\end{align}
where \eqref{eq:use_exp} follows from
\begin{align}
1-e^{-x} \leq x,& \qquad \forall x>0, \nonumber\\
\int_{x}^{\infty} \frac{1}{z}e^{-z} ~ \rmd z  \leq e^{-x}\ln (1+\frac{1}{x})\leq \ln (1+\frac{1}{x}),& \qquad \forall x>0, \label{eq:exponential}
\end{align}
and \eqref{eq:exponential} is an established bound (see \cite[5.1.20]{Abramowitz1964}).

Substituting \eqref{eq:small2} into \eqref{eq:big2}, we have
\[
\bbP_{\rmG_1\rmG_2}\left[\left|a^TXb\right|\leq \delta\right] \leq \frac{\delta}{\varepsilon\sigma_1\sigma_2}\left(1+\ln\left(1+\frac{E\sigma_1\sigma_2}{\delta}\right)\right),
\]
thus completing the proof.
\end{proof}

\begin{lemma}\label{lem:concentration3}
Suppose $a\in\bbR^{n}$ is a random vector following uniform distribution on $R\calB_{n}$. If symmetric matrix $X\in\bbR^{n\times n}$ satisfies $\norm{X}_2\geq \varepsilon$, then
\[
\bbP_{\rmU}\left[\left|a^TX a\right|\leq \delta \right] \leq \frac{2\sqrt{2\delta} \cdot V_{n-1}}{\sqrt{\varepsilon} R\cdot V_n}.
\]
\end{lemma}
\begin{proof}
Suppose the eigendecomposition of symmetric matrix $X$ is $X = U\Lambda U^T$, where $U$ is an orthogonal matrix, and $\Lambda$ is a diagonal matrix whose diagonal entries are the eigenvalues of $X$. Suppose $\Lambda^{(1,1)}\geq\Lambda^{(2,2)}\geq \cdots\geq \Lambda^{(n,n)}$, then $\max\{\Lambda^{(1,1)},-\Lambda^{(n,n)}\} = \norm{X}_2\geq \varepsilon$. Without loss of generality, let $\Lambda^{(1,1)} = \norm{X}_2$.

Let $\tilde{a} \defeq U^Ta$, then $\tilde{a}$ also follows the uniform distribution on $R\calB_n$. Therefore,
\begin{align}
& \bbP_{\rmU}\left[\left|a^TXa\right|\leq \delta\right] \nonumber\\
= & \bbP_{\rmU}\left[\left|\tilde{a}^T\Lambda\tilde{a}\right|\leq \delta\right] \nonumber\\
= & \frac{\int_{R\calB_n} ~\rmd\tilde{a} ~\ind{|\tilde{a}^T\Lambda\tilde{a}|\leq \delta} }{\int_{R\calB_n} ~\rmd\tilde{a}} \nonumber\\
=& \frac{1}{R^nV_n} \int\limits_{R\calB_{n-1}}  ~\rmd\tilde{a}^{(2:n)} ~ \phi(\tilde{a}),  \label{eq:big3}
\end{align}
where
\begin{align}
\phi(\tilde{a}) =& \int_{-R}^{R} ~\rmd\tilde{a}^{(1)} ~\ind{|\tilde{a}^T\Lambda\tilde{a}|\leq \delta} \cdot \ind{|\tilde{a}^{(1)}|^2 \leq R^2-\norm{\tilde{a}^{(2:n)}}_2^2} \nonumber\\
\leq & \int_{-R}^{R} ~\rmd\tilde{a}^{(1)} \ind{\left|\left(\tilde{a}^{(1)}\right)^2+\frac{1}{\norm{X}_2}\tilde{a}^{(2:n)T}\Lambda^{(2:n,2:n)}\tilde{a}^{(2:n)}\right| \leq \frac{\delta}{\norm{X}_2}}             \nonumber\\
\leq & 2\sqrt{\frac{2\delta}{\norm{X}_2}} \nonumber \\
\leq & 2\sqrt{\frac{2\delta}{\varepsilon}}. \label{eq:small3}
\end{align}
Substituting \eqref{eq:small3} into \eqref{eq:big3}, we obtain
\[
\bbP_{\rmU}\left[\left|a^TXa\right|\leq \delta\right] \leq 2\sqrt{\frac{2\delta}{\varepsilon}}\cdot \frac{R^{n-1}V_{n-1}}{R^nV_n} = \frac{2\sqrt{2\delta} \cdot V_{n-1}}{\sqrt{\varepsilon} R\cdot V_n}.
\]
\end{proof}

\begin{lemma}\label{lem:concentration3_gaussian}
Suppose $a\in\bbR^{n}$ is a random vector whose entries are i.i.d. following a Gaussian distribution $N(0,\sigma^2)$. If symmetric matrix $X\in\bbR^{n\times n}$ satisfies $\norm{X}_2\geq \varepsilon$, then
\[
\bbP_{\rmG}\left[\left|a^TX a\right|\leq \delta \right] \leq \frac{2\sqrt{\delta}}{\sqrt{\pi\varepsilon}\sigma}.
\]
\end{lemma}
\begin{proof}
Similar to the proof of Lemma \ref{lem:concentration3}, we use the eigendecomposition $X = U\Lambda U^T$, and $\tilde{a} \defeq U^Ta$. Without loss of generality, let $\Lambda^{(1,1)} = \norm{X}_2\geq \varepsilon$.

Then we have
\begin{align}
& \bbP_{\rmG}\left[\left|a^TXa\right|\leq \delta\right] \nonumber\\
= & \bbP_{\rmG}\left[\left|\tilde{a}^T\Lambda\tilde{a}\right|\leq \delta\right] \nonumber\\
=& \int\limits_{\bbR^{n-1}}  ~\rmd\tilde{a}^{(2:n)} ~ \phi(\tilde{a})\cdot p_{\rmG^{(2:n)}}(\tilde{a}^{(2:n)}),  \label{eq:big4}
\end{align}
where
\begin{align}
\phi(\tilde{a}) =& \int_{-\infty}^{\infty} ~\rmd\tilde{a}^{(1)} ~\ind{|\tilde{a}^T\Lambda\tilde{a}|\leq \delta} \cdot p_{\rmG^{(1)}}(\tilde{a}^{(1)}) \nonumber\\
\leq & \frac{1}{\sqrt{2\pi}\sigma}\int_{-\infty}^{\infty} ~\rmd\tilde{a}^{(1)} \ind{\left|\left(\tilde{a}^{(1)}\right)^2+\frac{1}{\norm{X}_2}\tilde{a}^{(2:n)T}\Lambda^{(2:n,2:n)}\tilde{a}^{(2:n)}\right| \leq \frac{\delta}{\norm{X}_2}} \cdot e^{-\frac{(a^{(1)})^2}{2\sigma^2}}      \nonumber\\
\leq & \frac{1}{\sqrt{2\pi}\sigma}\times 2\sqrt{\frac{2\delta}{\norm{X}_2}} \nonumber \\
\leq & \frac{2\sqrt{\delta}}{\sqrt{\pi\varepsilon}\sigma}. \label{eq:small4}
\end{align}
Substituting \eqref{eq:small4} into \eqref{eq:big4}, we obtain
\[
\bbP_{\rmG}\left[\left|a^TXa\right|\leq \delta\right] \leq \frac{2\sqrt{\delta}}{\sqrt{\pi\varepsilon}\sigma}.
\]
\end{proof}


\end{document}